\documentclass[10pt,twocolumn,twoside,english]{IEEEtran}  % Comment this line out
\IEEEoverridecommandlockouts                              % This command is only
\usepackage{graphicx}
\usepackage{epstopdf}
\usepackage{hyperref}
\usepackage{amsmath,amsfonts, amssymb}
\usepackage{amsthm}
\newtheorem{problem}{Problem}
\newtheorem{theorem}{Theorem}
\newtheorem{corollary}{Corollary}
\newtheorem{lemma}{Lemma}
\newtheorem{remark}{Remark}
\newtheorem{definition}{Definition}
\newtheorem{example}{Example}
\newtheorem{algorithm}{Algorithm}
\newtheorem{proposition}{Proposition}

\usepackage{epsfig} % for postscript graphics files
\usepackage{psfrag}

\usepackage{times}
\usepackage[tight,footnotesize]{subfigure}
\usepackage{bibspacing}
\usepackage{verbatim}
\usepackage{babel}
\usepackage{listings}
\usepackage{comment} 
\usepackage{multirow}
\usepackage{algorithmic}
\usepackage{algorithm}
\usepackage[tight]{subfigure}
\usepackage{mathtools}
\usepackage{color}
\allowdisplaybreaks
\begin{document}
\title{\LARGE \bf  Coding Schemes for Securing Cyber-Physical Systems Against Stealthy Data Injection Attacks}

\author{Fei Miao,~\IEEEmembership{Student Member,~IEEE,}
 Quanyan Zhu,~\IEEEmembership{Member,~IEEE,}
 Miroslav Pajic,~\IEEEmembership{Member,~IEEE,}
 and\\~George J. Pappas,~\IEEEmembership{Fellow,~IEEE}% <-this % stops a space
\thanks{This material is based on research sponsored by DARPA under agreement number FA8750-12-2-0247.  The U.S. Government is authorized to reproduce and distribute reprints for Governmental purposes notwithstanding any copyright notation thereon.  The views and conclusions contained herein are those of the authors and should not be interpreted as necessarily representing the official policies or endorsements, either expressed or implied, of DARPA or the U.S. Government. Part of he results in this work appeared at the 53rd Conference on Decision and Control, Los Angeles, CA, USA, December 2014~\cite{code_cdc14}. }% <-this % stops a space
\thanks{F.~Miao and G.~J.~Pappas are with the Department of Electrical and Systems Engineering, University of Pennsylvania, Philadelphia, PA, USA 19014.  Q.~Zhu is with the Department of Electrical and Computer Engineering, New York University, Brooklyn, NY, USA 11201. M.~Pajic is with the Department of Electrical and Computer Engineering, Duke University, Durham, NC, USA 27708. Email: \{\tt miaofei, pappasg\}@seas.upenn.edu,\{\tt quanyan.zhu\}@nyu.edu, \{\tt miroslav.pajic\}@duke.edu}.}
 
\maketitle

\begin{abstract}
\label{abstract}
This paper considers a method of coding the sensor outputs in order to detect stealthy false data injection attacks. An intelligent attacker can design a sequence of data injection to sensors and actuators that pass the state estimator and statistical fault detector, based on knowledge of the system parameters. To stay undetected, the injected data should increase the state estimation errors while keep the estimation residues small. We employ a coding matrix to change the original sensor outputs to increase the estimation residues under intelligent data injection attacks. This is a low cost method compared with encryption schemes over all sensor measurements in communication networks. We show the conditions of a feasible coding matrix under the assumption that the attacker does not have knowledge of the exact coding matrix. An algorithm is developed to compute a feasible coding matrix, and, we show that in general, multiple feasible coding matrices exist. To defend against attackers who estimates the coding matrix via sensor and actuator measurements, time-varying coding matrices are designed according to the detection requirements. A heuristic algorithm to decide the time length of updating a coding matrix is then proposed.
\iffalse
This method of transforming sensor outputs will lead to future work on using 
\fi
%This method of transforming sensor outputs also provides a future extension, how to apply a set of coding matrices, even the attacker is able to learn the coding matrix.
%such that the attacker only has a probability to stay stealth for some time, even the attacker knows the set of coding matrices.  
%inject data corresponding to the current transform matrix the system is applying. Thus the data injection attack is not stealth, 
%Using a linear system with control-cost optimal (but nonsecure) and secure (but cost-suboptimal) controllers in presence of replay attacks as an example, 
%Using examples of , we compare the 
\end{abstract}

\section{Introduction}
\label{sec:intro}
Cyber-physical systems (CPSs) integrate computation and communications to interact with physical processes.   Many applications are considered as CPSs, including high confidence medical devices, energy conservation, environmental control, and safety critical infrastructures--such as water supply systems, electric power, and communication systems~\cite{cps}. Therefore, security is a critical aspect of these systems, and CPSs involve additional challenges in control layer. The problem of secure control is defined, and reasons for mechanisms of information security, sensor network security alone are not sufficient for the security of CPSs are analyzed~\cite{secure_control}. 
The key challenges of CPSs securities are summarized in~\cite{secure_challenge}. 
%, it is not sufficient to focus on traditional methods of protect information in computer security perspective. 
%we should be able to design 

Novel attack-detection algorithms in cyber security area can be designed, by understanding how attacks affect state estimation and control of the system. Two algorithms to maximize the utility of encrypted devices placed to increase system security are proposed to reduce the cost of communication cost in power grids~\cite{sa_protection}. Tools are developed to protect state-estimation components from stealthy attacks from an intelligent attacker with a partial model of the system~\cite{cs_se}.
%Supervisory Control and Data Acquisition (SCADA) systems are developed.
%Since a large number of measurements are sent over unencrypted communication channels in power grids, the authors developed 

Researchers have explored fault detection, isolation and reconfiguration (FDIR) methods to ensure systems' safety and robustness~\cite{survey_fault}. Although active techniques have been designed to tackle various types of attacks, fundamental limitations still exist~\cite{limit_activedetection}. With a limited number of sensor and actuator compromised by the attacker, i.e., some elements of the injection vector is restricted to be zero, resilient state estimators have been designed by previous work. Fawzi et al.\ propose estimation and control schemes of noise free linear systems~\cite{est-control}. Pajic et al.\ present a robust state estimation method in presence of attacks to no more than half of the sensors for systems with noise and modeling errors~\cite{arse}. In contrast, we examine a different case where the attacker can inject an arbitrary vector to the communication between sensors and the estimator/detector/controller component, thus no element of the injection vector is constrained to be zero. 

%Because cyber attacks are neither physically dangerous nor constrained by geography, a risk-averse adversary prefers to launch cyber attacks~\cite{secure_control}. 
The monitoring system can detect malicious behaviors in general. Coding and decoding schemes to estimate the state of a scalar stable stochastic linear system with noisy measurements are designed in~\cite{Dey_estcode}. A distributed methodology for detecting and isolating multiple sensor faults in interconnected CPS is proposed in~\cite{Reppa_fd}.  A class of false data injection attacks against state estimators in power grid is analyzed in~\cite{fdi_PowerGrids}. Sequential detection techniques of sensor networks are discussed in~\cite{Nay_sd}. Miao et al.\ design stochastic game approaches for replay attacks detections~\cite{game_replay} and secure control of CPSs~\cite{Miao_game}. 

However, with knowledge of the system model, an intelligent cyber attacker is able to carefully design a data injection sequence, such that the state estimation error increases without triggering the alarm of the monitor~\cite{false_injection},~\cite{accstealth}. Manandhar et al.\ design the Euclidean detector to overcome the limitation of $\chi^2$ detector for fault detection in smart grid~\cite{Mana_kffd}. However, the design of Euclidean detector is based on the voltage signal model of smart grid and whether it works for a general linear system model has not been shown yet. In this work, we consider the detection problem of false data injection attacks for a general linear system model. To address the computational overhead of encryptions on embedded architectures~\cite{encrypt_sensor}, we propose an alternative low cost method to code the sensor measurements for detection. With the coding scheme, no additional detector is required for the system to detect stealthy data injected by an attacker with the knowledge of system model. Compared with error-correcting coding schemes~\cite{correct_code1977}, the sensor outputs coding approaches proposed in this work aim to change the value transmitted over the communication channel instead of correcting errors on bit level. Moreover, the coding scheme proposed in this work does not require additional bits for each plaintext message of the sensor measurements, while an encryption method introduces communication overhead for each sensor message transmitted in the communication channel~\cite{encrypt_key}. We assume that the coding matrix is distributed between sensors and the estimator/detector of the system correctly like an secret encryption key~\cite{encrypt_sn}, and measurement of individual sensor is not corrupted before coded. With the coding matrix, the values sent over the communication channel are changed, without additional bits for encryption overhead~\cite{correct_code1977}, and the scheme is low-cost compared with the scheme of encrypting all sensor outputs.

The contributions of this work are summarized as follows:
\begin{enumerate}
\item The main contribution of this work is a low cost method of coding sensor outputs to detect stealthy false data injection attacks.
We show that the system can detect the original stealthy sensor injections by coding the sensor outputs according to certain conditions. % even if the attacker knows the system model without the coding scheme, %system dynamics and the original sensor setup, %and a heuristic method
\item We also design an algorithm to compute such coding matrices, and show that in general, multiple feasible coding matrices exist. 
\item When the attacker can estimate the coding scheme according to several measurements of sensor and actuator values, we show that it is difficult to get the exact coding matrix in general. Moreover, in this case, the system can either change a new coding matrix or randomly use a set of coding matrices within a time length before the attacker has enough measurements for a good estimation. We design a heuristic algorithm to decide the time length of updating a coding matrix. 
\end{enumerate}
% the coding approach saves encryption cost compared with encrypting all sensor outputs. 
%Then there is only a probability that the attacker is injecting a stealth sequence according to the coding matrix the system is applying. 
%In practice this will work since the attacker always needs some time to figure out a new linear combination of sensor outputs. 
%\FM{be careful about the contribution part.}

%summarize related journal paper from TCNS

The paper is organized as follows. In Section~\ref{sec:prob} we describe the system and attack models. The conditions that a feasible coding matrix should satisfy are presented in Section~\ref{sec:stealth}. An algorithm to find a feasible coding matrix based on rotation matrix is developed in Section~\ref{algorithm}. A time-varying coding scheme is designed in Section~\ref{Sig_t}. Section~\ref{sec:simulation} shows illustrative examples. Conclusions are given in Section~\ref{sec:conclusion}.

\section{System And Attack Model}
\label{sec:prob}
\begin{figure}[b!]
\vspace{-8pt}
\centering
\includegraphics [width=0.38\textwidth]{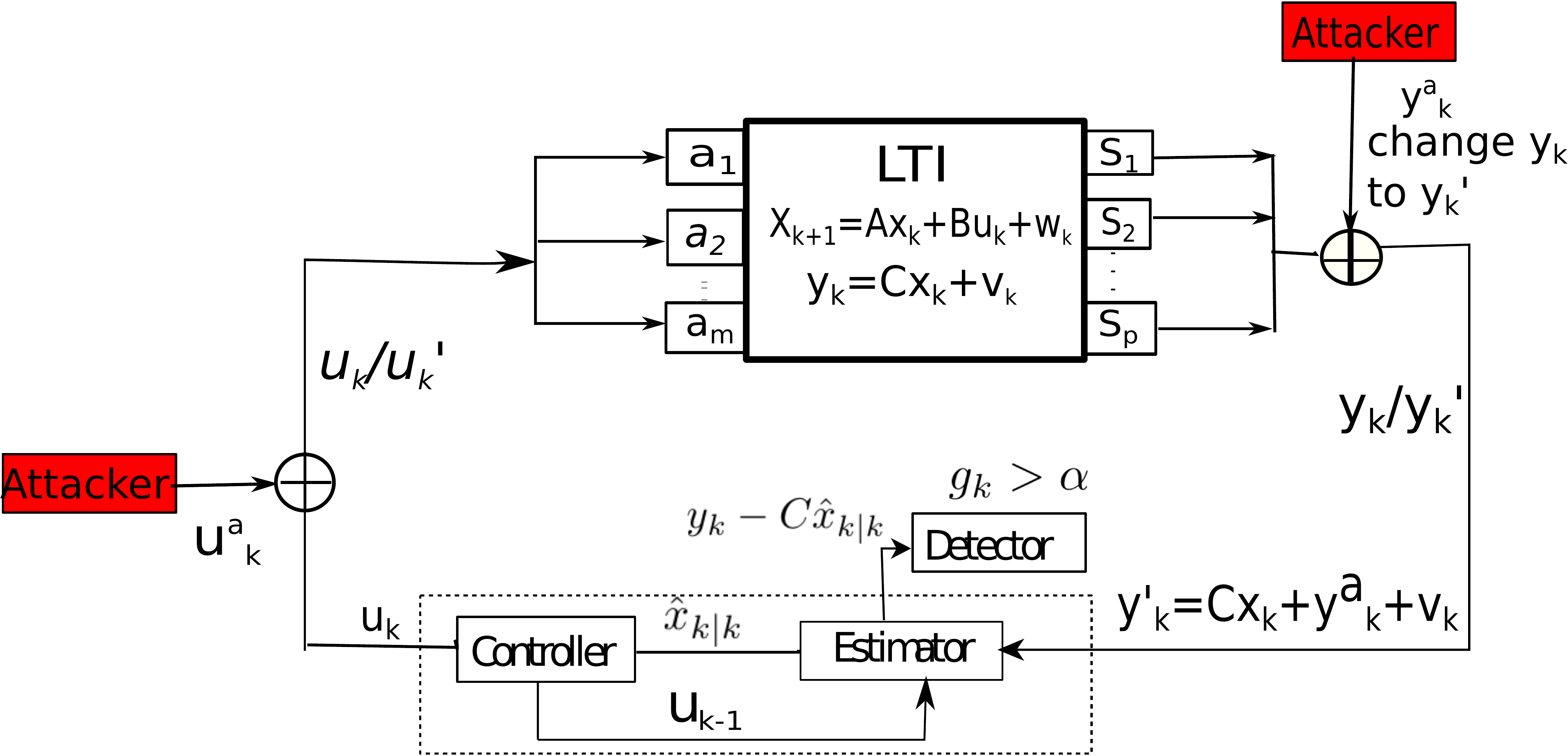}
\vspace{-10pt}
\caption{System diagram, where the system is equipped with an estimator, a detector and a controller. The attacker can inject arbitrary false data vector $y^a_k$ to sensor outputs and $u^a_k$ to actuator inputs.}
\label{fig_sys}
\end{figure}
We will introduce a discrete-time linear time-invariant (LTI) system model, a data injection attack model, and the attacked system model in this section.  
The system architecture is shown in Figure~\ref{fig_sys}. 
%Assume the system is equipped with a Kalman filter and a $\chi^2$ detector for fault detection. 
\iffalse The choice of linear controller does not affect the design of detection method in the following sections, and without loss of generality we can consider the optimal LQG controller as an example in the system diagram.  \fi
\subsection{Linear system model}
\label{sys_model}
 Assume that the CPS is composed of a discrete time LTI system with the following form:
%\subsubsection{Normal plant model}:
\begin{align}
\begin{split}
x_{k+1}=Ax_{k}+Bu_{k}+w_{k},\quad
y_{k}=Cx_{k}+v_{k},
\end{split}
\label{system}
\end{align}
where $x_{k} \in \mathbb{R}^{n}$ is the system state vector, $u_{k} \in \mathbb{R}^{m}$ is the control input, and $y_{k} \in \mathbb{R}^{p}$ is the sensor observations at time $k$. We do not have specific restrictions for the linear control input $u_k$ here, since the choice of a linear controller does not affect the detection of false data injection, and we will explain the reason later. We assume that $w_{k}\sim N(0,Q)$ and $v_{k}\sim N(0,R)$ are identical independent (i.i.d.) Gaussian noises. %and initial state of the system satisfies $x_{0}\sim N(0, \Theta)$. 

The optimal Kalman filter used to estimate state $\hat{x}_{k|k}$ is:
\begin{align*}
&\hat{x}_{0|-1}=0,\ \ P_{0|-1}=\Theta,
P_{k+1|k}=AP_kA^T+Q,\\
&K_{k+1}=P_{k+1|k}C^T(CP_{k+1|k}C^T+R)^{-1},\\
&P_{k+1}=(I-K_{k+1}C)P_{k+1|k},\\
&z_{k+1}=y_{k+1}-C(A\hat{x}_ {k}+Bu_k),\\
&\hat{x}_{k+1|k}=A\hat{x}_k+Bu_k,\quad
\hat{x}_{k+1}=\hat{x}_{k+1|k}+K_kz_{k+1}.
\end{align*}
Under the assumption that $(A,B)$ is stabilizable, $(A,C)$ is detectable, we get a steady state Kalman filter, with the error covariance matrix $P$ and Kalman gain matrix $K$:
\begin{align*}
P \triangleq \lim \limits_{k \to \infty} P_{k|k-1}, K\triangleq PC^T(CPC^T+R)^{-1}.
\end{align*}

% Here $z_{k+1}$ is the estimation residue.
%Then the state estimation at time $k$ is:
%\begin{align*}
%&\hat{x}_{k+1}=\hat{x}_{k+1|k}+Kz_{k+1}, 
%\end{align*}
%where $z_{k+1}=y_{k+1}-C(A\hat{x}_ {k}+Bu_k)$ 
Without attacks, the estimation residue $z_{k}$ follows a Gaussian distribution $N\sim (0, CPC^{T}+R)$. 
Define the quantities $g_{k}$ as 
%\begin{align*}
$g_{k}=z_{k}^{T}P^{-1}z_{k}$,
%\end{align*}
where $P$ is the error covariance matrix of Kalman filter, then $g_k$ satisfies a $\chi^{2}$ distribution with $p$ degrees of freedom. A $\chi^{2}$ failure detector considers the standardized residue sequence 
$\eta_{k}= P^{-\frac{1}{2}}z_{k}$
for a monitoring system, and assumes that there exists a $\delta_{\eta}$ such that
%\begin{align*}
$\lim_{k\to \infty} \|E{\eta_{k}}\| \leq \delta_{\eta}.$
%\end{align*} 
We denote $\alpha$ as the threshold for detecting a fault, meaning that the alarm is triggered when $g_{k}> \alpha.$
%\begin{align*}

%\end{align*}

\subsection{False data injection attack model}
%In this paper, we assume that actuators are secure and only consider the case of sensor attacks.  
%only consider attacks on sensor and actuator/controller, not include physical attack. 
The system model under sensor data injection attack is described as~\eqref{attackmodel}       %is also consistent with the math formulation in paper %\cite{attackmodel}):
\begin{align}
\begin{split}
x'_{k+1}&=Ax'_{k}+B(u'_k+u^a_k)+w_{k}, \\
y'_{k}&=Cx'_{k}+y^{a}_{k}+v_{k},
\end{split}
\label{attackmodel}
\end{align}
where $y^a_k \in \mathbb{R}^p$, $u^a_k \in \mathbb{R}^m$ are arbitrary vectors injected to sensor outputs, actuator inputs by the attacker at time $k$ respectively. When $u^a_k=0$, only sensor values are changed by the attacker.
Assume the adversary has knowledge of the system model described in Section~\ref{sys_model}, and is able to inject data over communication network between sensors and the estimator/detector/controller. 
\iffalse
In particular, $y^a_k \in \mathbb{R}^m$ is an arbitrary vector designed by the attacker and no element is restricted to be zero.
\footnote{An attack model restricting attacks on some sensors and actuators is:
\begin{align*}
\begin{split}
x'_{k+1}=Ax'_{k}+Bu'_{k}+B^{a}u^{a}_{k}+w_{k},
y'_{k}=Cx'_{k}+C^a y^{a}_{k}+v_{k},
\end{split}
%\label{attackmodel}
\end{align*}
where $C^a=diag(c^a_{1},\dots, c^a_{p})$, $c^a_{i}=1$ if the $i$-th sensor can be compromised, $y^{a}_{k}$ is the bias introduced by the attacker. Similarly, $B^{a}=diag(b^a_{1},\dots, b^a_{m})$, $b^a_j=1$ if the attacker inject to the $j$-th actuator.}
\fi

Without attack, according to the system dynamics and the definition of Kalman filter, the estimation error is 
\begin{align*}
&e_{k} \triangleq x_k-\hat{x}_k,\\
&e_{k+1} = (A-KCA) e_{k}-Kv_k+ (I-KC)w_k.
\end{align*}
When matrix $(A-KCA)$ is stable and $\mathbb{E}w_k=\mathbb{E}v_k=0$, the expectation of estimation error converges to $0$ with a static Kalman filter, i.e., $\lim_{k\to\infty} \mathbb{E}[e_k]\to 0$. Meanwhile, the residual $z_k$ stays in the subspace that does not trigger the alarm with a high probability. 

To illustrate how the sensor injection sequence $y^a_k$ will affect the estimation and monitoring system, we examine how the estimation error and residue will change with $y^a_k$. Denote the estimation residuals of attacked system as
\begin{align*}
z'_{k}=y'_{k+1}-C(A\hat{x}'_ {k}+Bu'_k),
\end{align*}
 where $\hat{x}'_{k}$ is the state estimation of the compromised system. Similarly, we define the estimation error under attack as
 \begin{align*}
 e'_{k} \triangleq x'_k-\hat{x}'_k,
 \end{align*}
 The probability that the sensor injection sequence $y^a_k,\ k=0,1,\dots$ is detectable is given by
 \begin{align*}
 Pr(g'_{k}=(z'_{k})^{T}P^{-1}z'_{k} >\alpha \ \text{for any}\  k).
 \end{align*} 
 
% When only sensors values are injected, i.e., $u^a_k=0, k=0, 1\dots,$, t
 The difference between the normal and the compromised systems can be captured by:
\begin{align}
%\begin{split}
%&, \quad \\
\Delta e_{k} \triangleq e'_{k}-e_{k}, \quad \Delta z_{k} \triangleq z'_{k}-z_{k}.
%\end{split}
\label{e}
\end{align}
The dynamics of the above difference vectors satisfy
%\footnotesize
\begin{align}
\begin{split}
%\Delta \tilde{x}_{k+1}=&\begin{bmatrix}A+BL&-BL\\
    %                                                0&A-KCA\\
        %                                            \end{bmatrix}\Delta \tilde{x}_{k} +
            %                                        \begin{bmatrix}0\\
                %                                                                   -Ky^{a}_{k+1}\end{bmatrix} \\ \fi
&\Delta e_{k+1} = (A-KCA) \Delta e_{k} - Ky^{a}_{k+1}+(B-KCB)u_k^a,\\
&\Delta z_{k+1}= CA\Delta e_k+ y^a_k+CBu_k^a,
%\Delta z_{k+1}=&\begin{bmatrix}0&CA\end{bmatrix}\Delta \tilde x_{k} +\begin{bmatrix}0 & I\end{bmatrix} \begin{bmatrix}0 \\y^{a}_{k+1}\end{bmatrix}.                                                                                                                                       
\end{split}    
\label{delta_z}                                                                
\end{align}
%\normalsize
Hence the difference vectors between normal and compromised systems, $\Delta z_{k}(y^{a}, u^a), \Delta e_{k}(y^{a},u^a)$, are functions of the injection sequences $y^{a}\triangleq (y^{a}_0, y^{a}_1,\dots)$, $u^a \triangleq(u^a_0, u^a_1,\dots)$. 
To simplify the notations, we concisely denote these vectors as $\Delta z_{k}, \Delta e_{k}$, respectively. 

The objectives of the attacker include increasing the estimation error $e'_k$ without triggering the alarm, and destabilizing the system with infinite state estimation error $e'_k$ in the long run. Note that these types of attacks on control systems have been illustrated in the recent years. For instance, the estimated trajectories of Unmanned Ground Vehicle (UGV)~\cite{arse} and Unmanned Aerial Vehicle (UAV) navigation systems~\cite{accstealth} under stealthy data injection attacks (e.g., by GPS spoofing) deviate from the actual trajectories of the autonomous vehicles before being detected. Thus the attacker's objective is equivalent to increasing $\|\Delta e_k\|_2$ (the difference between estimation error of the normal and compromised systems) to infinity without increasing $\|\Delta z_k\|_2$ much as time goes by. Since computing the detecting statistic of compromised system $g'_{k}$ is to integrate a Gaussian distribution on an ellipsoid,  the stealthy requirement can be  approximated by keeping $\|z'_k\|_2$ small. Residues of the normal system $z_k$ are bounded, and the attacker should keep the change of residues bounded make the injection stealthy. It means the following inequality should hold
 \begin{align}
\|\Delta z_{k}\|_2 \leq M,
\label{residue} 
\end{align}
where $M$ is a residue norm change threshold designed by the attacker. The compromised estimation residue should be close to that of the normal system, to deceive the monitoring system. %Without loss of generality, we can choose $\beta =1$. 
 %For a stealth data injection sequence, $\|\Delta z_{k}\|_2$ is bounded.
 \footnote{The relation between the scale or norm of the injection sequence and the alarm trigger threshold $\alpha$ is shown in Theorem 1 in~\cite{accstealth}.}
% \subsection{Only sensors are attacked}
%%%%%%%%%%%%%%%%
%%%%%%%%%%%%%%%%%%
%For a special case when the attacker has the ability to inject false data to the full observation $y_k$, i.e. , 
%%%%%%%%%%%%%%%%%%
%%%%%%%%%%%
When $y^{a}_k$ can be an arbitrary vector, a necessary and sufficient condition for a stealthy injection $y^a_k$ that can increase $\|e'_k\|_2$, $\|x'_k\|_2$ to infinity while keep $\|z'_k\|_2$, $\|\Delta z_{k}\|_2$ bounded is derived in~\cite{accstealth},~\cite{false_injection}. The condition that $Cv \in span (I)$, i.e., there exists $y^{*}$ satisfying $y^{*} =Cv$ is always satisfied by the attack model~\eqref{attackmodel}. Hence, we have the following proposition. %(they proved equivalent results with different statements):
\begin{proposition}%[\cite{accstealth,false_injection}]
There exists a stealthy sequence $y^a_k, k=0, 1, \dots,$ given the attacked system model~\eqref{attackmodel}, if and only if matrix $A$ has an unstable eigenvalue $\lambda$ and the corresponding eigenvector $v$, such that
$v \in span (Q_{oa})$, where $Q_{oa}$ is the controllability matrix associated with the pair $(A-KCA, K)$.
\label{stealth_fi}
\end{proposition}

\section{Coding Sensor Outputs For Detecting Stealth Sensor Data Injection}
\label{sec:stealth}
Existing statistical detectors, active monitor schemes (design some additive control input $u^d_k$) and fault detection filters have limitations, that even actuators are not compromised, they cannot detect stealthy sensor data injection attacks. It is necessary to design some inexpensive techniques to compensate for the vulnerability of the system under intelligent sensor data injection attacks. It has been shown that by only compromising sensors, attackers can induce infinite estimation error without being detected under monitoring systems like a 
$\chi^2$ detector~\cite{accstealth}. Therefore, we first discuss the case of stealthy sensor false data injection attacks in this section.

\subsection{Limitations of existing approaches}%{active monitor approach and fault detection filter}
\textbf{The limitation of active monitor approach}: %there is only sensor data injection and
Under the assumption that actuators work appropriately for the attacked system~\eqref{attackmodel}, the challenge here is whether adding $u^d_k$ to the pre-designed linear control input $u_k$ (such as optimal LQG control) can help to detect stealthy sensor data injections. 
For instance, consider a new control input 
\begin{align}
\tilde{u}_k=u_k+u^d_k,
\label{ud}
\vspace{-8pt}
\end{align}
where $u^d_k$ is some random authentication signal or a constant value. It is worth noting that active monitor approaches do not help for detecting sensor data injection attacks described in model~\eqref{attackmodel} . 
\begin{lemma}
There exists no active monitor in the form~\eqref{ud} that can increase the detection probability of a stealthy sensor data injection sequence, for the system~\eqref{system} equipped with a Kalmen Filter and a $\chi^2$ detector.
\end{lemma}
\begin{proof}
We denote the difference between estimation residual and estimation error of the normal and compromised system for the system with the controller~\eqref{ud} as $\Delta \tilde{z}_{k+1}$ and $\Delta \tilde{e}_{k+1}$, respectively. By the definition of $\Delta \tilde{z}_{k+1}$ and $\Delta \tilde{e}_{k+1}$ and a similarly calculation process to get~\eqref{delta_z}, we have 
$\Delta \tilde{e}_{k}=\Delta {e}_{k}$, and $\Delta \tilde{z}_{k+1}=CA \Delta \tilde{e}_k+y^a_{k+1}+CBu^a_k$.
%\footnotesize
%\begin{align*}
%\begin{split}
%\Delta \tilde{z}_{k+1}= &\tilde{z}'_{k+1}-\tilde{z}_{k+1} \\
%                                     =&\tilde{y}'_{k+1}-C(A\tilde{x}'_k +B\tilde{u}'_k)-\tilde{y}_{k+1}
 %                                      +C(A\tilde{x}_k +B\tilde{u}_k)\\
  %                                  =&CA \Delta \tilde{e}_k+y^a_{k+1}.
%\end{split} 
%\label{actu}                                   
%\end{align*}
%%%%%%%%%%%%%%%%%%%%%%%%%
%%%%%%%%%%%%%%%%%%%%%%%%%
%\normalsize
Any additional control input $u^d_k$ will be eliminated by the deduction of $\tilde{z}_{k+1}$ and $\tilde{z}'_{k+1}$ to get $\Delta z_{k+1}$. 
The active control input does not increase the norm of $\Delta \tilde{z}_{k+1}$ compared with $\Delta z_{k+1}$, which means there exists no linear form of $\tilde{u}_k$ as described above that can increase $\|\Delta z_{k+1}\|_2$ under $y^a_k$ for the system~\eqref{attackmodel}. 
\end{proof}
%%%%%%%%%%%%%%%%%%%%%%%%%%
%%%%%%%%%%%%%%%%%%%%%%%%%%%%

The limitations of active monitors for a unified LTI model are explained in Theorem 4.7 of~\cite{limit_activedetection}
\footnote{A different case when adding exogenous Gaussian distribution control input can detect replay attacks is discussed in~\cite{replay}.}.
%{A different case when adding exogenous Gaussian distribution control input can detect the attack is discussed in~\cite{replay}--for replay attack active monitor approach works.} 
From this perspective, different linear controllers are equivalent under stealth sensor data injection attacks, and we do not restrict the controller model for designing our detection techniques. 

\textbf{The limitation of fault detection filter}:
Besides Kalman filter, observer-based fault detection filters for LTI systems with unknown error have been developed. The design requirements usually include robustness to unknown inputs and sensitivity to faults. 
\iffalse therefore, we examine the usual type of fault detection filter for the stealth described by~Theorem~\ref{stealth_fi}.  \fi
Such filters generate a different residue from $z_k$ of Kalman filter. Consider the following form of residual generator and residual evaluator (including a threshold and a decision logic unit, see~\cite{fd_continuous} for details)~\cite{fd_continuous}:
%\footnotesize
\begin{align}
\begin{split}
%&\hat{\mathbf{x}}_{k+1}= (\mathbf{A}-\mathbf{HC})\hat{\mathbf{x}}_k+\mathbf{Bu}_k + \mathbf{Hy}_k,\\
\hat{x}_{k+1}&= A\hat{x}_k+Bu_k + H(y_k-\hat{y}_k),\\
\hat{y}_k&=C\hat{x}_k, \quad
r_k=V(y_k-\hat{y}_k),
\end{split}
\label{fd_filter}
\end{align}
%\normalsize
where $\hat{x}_k \in \mathbb{R}^n$ and $\hat{y}_k \in \mathbb{R}^p$ represent the state and output estimation vectors, respectively, and $r_k$ is the residual signal.
This fault detector shares the same limitation with Kalman filter, i.e., the intelligent sensor data injection attack is stealth for the filter described as~\eqref{fd_filter}, since the residue is still observer based difference between $y_k$ and $\hat{y}_k$.  
\iffalse The injection sequence designed for with $Cv$ also works. \fi

\subsection{Coding sensor outputs to detect stealth data injection}
Since existing monitoring systems cannot detect intelligent false data injection attacks, and encryption method has a constraint of significant computation overhead, we propose a design of \textit{coding the sensor outputs} to detect stealth sensor data injection attacks. An intelligent attacker designs the sequence $y^a_k$ carefully to keep the change of residue $\|\Delta z_k\|_2 \leq M$, where $M$ is a constant. Thus, the objective of a detecting approach is equivalent to increasing $\|\Delta z_k\|_2$ as fast as possible under a stealthy data injection sequence, and $\|\Delta z_k\|_2$ should increase to infinity as time goes to infinity.
%\\\textbf{Coding for Observer}: 
\begin{figure}[b!]
\centering
\includegraphics[width=0.42\textwidth]{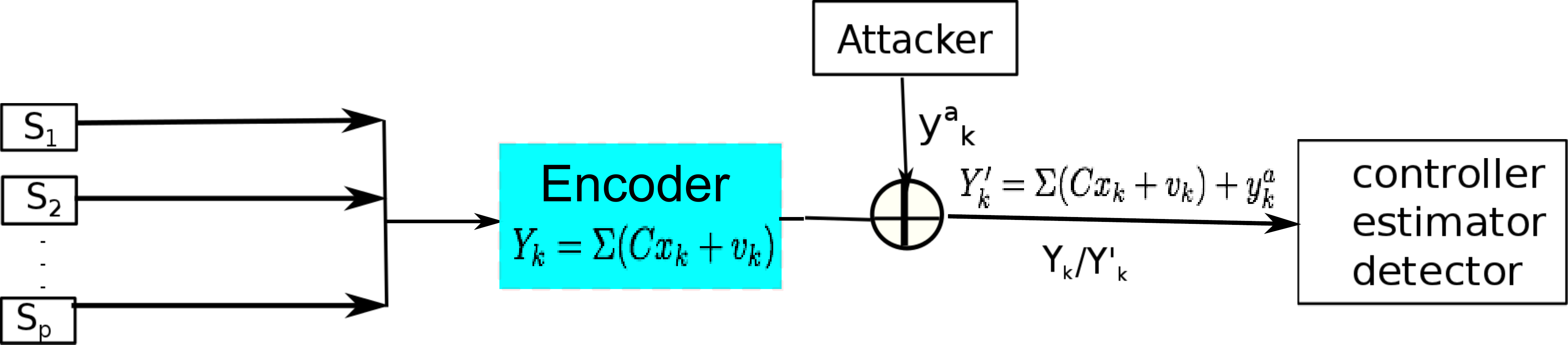}
\vspace{-8pt}
\caption{System diagram when coding sensor outputs with a matrix $\Sigma$ that satisfies the conditions of Theorem~\ref{code}. The attacker can inject arbitrary false data vector $y^a_k$ to sensor outputs.}
\label{fig_trans}
\end{figure}

The necessary and sufficient conditions for stealth false sensor data injection in~Corollary~\ref{stealth_fi} assume that the attacker knows $(A,B,C,K)$. Parameters $A$ and $B$ are related to physical dynamics that may not be altered,  while $C$ is related to the sensor measurements, corresponding specific physical states. 
Without changing the physical setup, we still can manipulate the sensor outputs. To violate the attacker's design,
we consider the method of transforming sensor outputs as shown in Figure~\ref{fig_trans}--instead of sending the output vector $\mathbf{y}_k=Cx_k+v_k$ to the estimator/controller/detector, sensors transmit the value 
%\FM{in one package or not?}\iffalse in one package: \fi
\begin{align}
Y_k =\Sigma (Cx_k +v_{k}),  C\in \mathbb{R}^{p\times n},
\label{sig_y}
\end{align} 
where $\Sigma \in \mathbb{R}^{p\times p}$ is an invertible matrix. We assume that the measurement of individual sensor is not corrupted yet before coding, and injection sequence appears in the communication between sensors to the estimator/controller/detector. One can think of $\Sigma$ as an inexpensive code, and compare $\Sigma$ with an encryption key. By encrypting only the coding matrix channel once, the coding approach saves encryption cost compared with encrypting all sensor outputs for every time $k$. 

We assume that the attacker does not know the matrix $\Sigma$ at least before estimating the matrix based on knowledge of matrices $(A,B,C,K)$ and sensor and actuator values, since the coding matrix $\Sigma$ is not fixed by the physical model of the system and can be time-varying, calculated in polynomial time when a new coding matrix is needed (an algorithm will be proposed in the next section). We assume that the attacker cannot access the coding matrix directly if he/she only applies eavesdropping techniques to the unencrypted communication channel and the process of distributing $\Sigma$ is protected. We will propose a time-varying coding scheme later in this work. When the attacker has designed a sequence of stealthy attack signal $y^{a}_k$ for the original system without the knowledge of the coding matrix $\Sigma$, the false sensor value after coding changes to:
\begin{align}
Y'_k =\Sigma Cx_k + y^{a}_{k}+\Sigma v_{k}.
\label{sig_ya}
\end{align}

Since $\Sigma$ is an invertible matrix, when the state estimator receives $Y_k$ or $Y'_k$, the encoded packet is decoded as
\begin{align}
\tilde{y}'_k= \Sigma^{-1}Y'_k=y_k+\Sigma^{-1}y^a_k,
\label{y_decode}
\end{align}
And we still use the same Kalman filter and $\chi^2$ detector on the decoded sensor outputs. Similar as the definitions of $\Delta e_k$ and $\Delta z_k$~\eqref{delta_z} for sensor outputs before coding, we define $\Delta e'_k$ and $\Delta z'_k$ as the change of state estimation and residue for coded sensor outputs without attack~\eqref{sig_y} and under attack~\eqref{sig_ya}, respectively. With $\Sigma$, a stealth data injection designed for~\eqref{system} (with parameters $(A,B, C,K)$),  $\|\Delta z'_k\|_2$ increases to infinity as $k \to \infty$ under certain conditions.
In the following theorem, we show the sufficient conditions that $\Sigma$ should satisfy for any stealth sequence of $y^a_k, k=0,1,\dots$ that satisfies Theorem~\ref{stealth_fi}.
%%%%%%%%%%%%%%%%%%%%%%%%%%%%%%%%

\begin{theorem}
Given an attacked system model~\eqref{attackmodel}, assume that $(A,C)$ is detectable, $u^a_k=0,$ and the attacker designs a sequence of sensor data injection $y^a_k, k=0,1,\dots$, based on one unstable eigenvector $v \in span (Q_{oa})$, where $Q_{oa}$ is the controllability matrix associated with the pair $(A-KCA, K)$. If there exists an invertible matrix $\Sigma$, and the direction of $\Sigma Cv$ is not the same with that of $Cv$, i.e.,  %such that $(A,\Sigma C)$ is detectable
\begin{align}
%\Sigma\mathbf{C}\neq 0,
\frac{(Cv)'\Sigma Cv}{\|\Sigma Cv\|_2\|Cv\|_2} \neq  1, 
\label{Sigma_c}
\end{align}
then after injecting $y^a_k$ the estimation residue change satisfies $\lim_{k\to\infty}\|\Delta z'_k\|_2 \to \infty$, by coding sensor outputs~\eqref{sig_y} with $\Sigma$. 
%%%%%%%%%%%%%%%%%
\label{code}
\end{theorem}

\begin{proof}
Given a system under data injection attacks as~\eqref{attackmodel}, we assume that the system has one unstable eigenvector $v$ with corresponding eigenvalue $\lambda$. According to the definition in equation~\eqref{e}, the dynamics of $\Delta e_{k}, \Delta z_k$ satisfy~\eqref{delta_z} with $u^a_k=0$.
For coded sensor outputs~\eqref{sig_ya}, after decoding %dynamics of $\Delta e'_k, \Delta z'_k$ satisfy:
\begin{align}
\begin{split}
\Delta e'_{k+1}&=(A-KCA )\Delta e'_{k} - K\Sigma^{-1} y^{a}_{k+1},\\
\Delta z'_{k+1}&=CA\Delta e'_k+ \Sigma^{-1}y^a_{k+1},
%\tilde{A}=A-K'\Sigma CA,
%&\Delta \mathbf{e}'_{n-1}=\mathbf{v}+\Delta \mathbf{v},
\end{split}
\label{ez}
\end{align}
The proof of \textit{Theorem $1$} in~\cite{accstealth} shows that under a stealthy sensor data injection sequence, the only component of $\Delta e_k$ that goes to infinity eventually depends on the unstable eigenvector, denoted as
%\begin{align}
$c_k v, \lim_{k \to \infty} c_k =\infty,$
%\label{infty}
%\end{align} 
and $\Delta e_k$ can be decomposed as
%\begin{align}
$\Delta e_k=c_k v+\epsilon_{1k}, \|\epsilon_{1k}\|_2 \leq M_1.$
%\end{align}

To keep $\Delta z_k$ bounded as $k \to \infty$, any stealthy injection sequence $y^a_k$ must satisfy
\begin{align}
y^a_{k+1}=- c_k\lambda Cv+ \epsilon_{2k}, \|\epsilon_{2k}\|_2 \leq M_2, k=0,1,2,\dots,
\label{ya}
\end{align}  
where $M_2$ is a constant such that $\|\Delta z_k\|_2 \leq M$ for all $k$. 

We assume that the attacker does not know $\Sigma$, and designs an injection sequence for the original system~\eqref{system} as described in~\eqref{ya}. Similarly as $\Delta e_k$, the only component of $\Delta e'_k$ that can goes to infinity is $c_k v$, since matrix $A$ is not changed by the coding matrix $\Sigma$. However, with any $y^a_k$ in~\eqref{ya}, $\Delta z'_k$ can be decomposed as 
\begin{align}
\Delta z'_k=c_k\lambda(Cv-\Sigma^{-1} Cv) + \epsilon_{3k}, k=0,1,2,\dots, 
\label{sig_z}
\end{align} 
where $\epsilon_{3k}$ is a bounded vector components of $\Delta z'_k$.  
When $\Sigma$ satisfies equation~\eqref{Sigma_c}, $ \Sigma Cv-Cv \neq 0$. With $c_k \to \infty$, $\|\Delta z'_k\| \to \infty$ as $k \to \infty$. 
\end{proof}
%%%%%%%%%%%%%%%%%%%%
%%%%%%%%%%%%%%%%%%%%
%%%%%%%%%%%%%%%%%%%%

We call a matrix $\Sigma$ that satisfies the conditions of Theorem~\ref{code} a feasible coding matrix. Theorem~\ref{code} shows that even the attacker knows system parameters $(A,B,C,K)$, without changing the physical structure or altering $A,B$, we can utilize the sensor data to get different residues for detecting.
Leveraging sensor outputs is the key reason to detect a stealth sensor data injection. 
%\FM{is this the correct description for $span{Cv}$?}
It is worth noting that here we do not constrain specific structure of the matrix $\Sigma$ besides conditions in Theorem~\ref{code}. For an LTI system, $\Sigma C$ is simply a linear transform of the original sensor measurement. When $A$ has several unstable eigenvectors satisfying Corollary~\ref{stealth_fi}, the following lemma extends the result of Theorem~\ref{code}.
\begin{lemma}
Given an attacked system~\eqref{attackmodel} with $(A,C)$ detectable and a set of unstable eigenvectors $v_1, \dots, v_u \in span (Q_{oa})$, where $Q_{oa}$ is the controllability matrix associated with the pair $(A-KCA, K)$, if $\Sigma$ is an invertible matrix, and %such that $(A,\Sigma C)$ is detectable
\begin{align}
\frac{(C\tilde{v})'\Sigma C\tilde{v}}{\|\Sigma C\tilde{v}\|_2\|C\tilde{v}\|_2} \neq  1,
\end{align}
for any linear combinations of $v_1, \dots, v_u$ -- $\tilde{v}$, then $\Sigma$ is a feasible coding matrix to  increase $\|\Delta z'_k\|_2$ for any stealth data injection to attacked system~\eqref{attackmodel}.
\label{code_lemma}
\end{lemma}
\begin{proof}
When matrix $A$ has a set of unstable eigenvectors $v_1,\dots, v_u$ with corresponding eigenvalues $\lambda_1,\dots, \lambda_u$, similar as the proof of Theorem~\ref{code}, a stealthy injection sequence takes the form
\begin{align*}
y^a_k= \sum_{i=1}^{u}c_{ik} \lambda_i C v_i + \epsilon_{2k}, \|\epsilon_{2k}\|_2 \leq M_2, k=0,1,2,\dots,
%\label{ya_}
\end{align*}
and the change of residual is defined as
\begin{align*}　
\Delta z'_k &=\sum_{i=1}^{u}c_{ik}\lambda_i(Cv_i-\Sigma^{-1} Cv_i) + \epsilon_{3k},\\
                    &=C(\sum_{i=1}^{u}c_{ik}\lambda_i v_i)-\Sigma^{-1}C(\sum_{i=1}^{u}c_{ik}\lambda_i v_i), k=0,1,2,\dots.
　%\label{sig_z}
\end{align*}

Hence, we consider $\tilde{v}=\sum_{i=1}^{u}c_{ik}\lambda_i v_i$ as a linear combination of all the unstable eigenvectors, the conclusion holds with the coding matrix $\Sigma$ satisfying all the constraints. 
\end{proof}
\begin{remark}
When the attacker is able to learn $\Sigma$ by analyzing sensor outputs and actuator inputs, the system can send a new $\Sigma$ before the attacker figures out the current applied coding matrix. The process of learning $\Sigma$ from the perspective of an attacker will be discussed in Section~\ref{Sig_t}.
\end{remark}
%%%%%%%%%%%%%%%%%%%%%%
%%%%%%%%%%%%%%%%%%%%%%
\iffalse
When $(A,\Sigma C)$ is detectable, there exists a steady state Kalman filter with parameter $K'$ for the coded system, and the corresponding fault detector. 
Hence, when the attacker designs a stealth injection sequence without knowledge of $\Sigma$, the system can detect it by increasing $\|\Delta z'_k\|_2$ with $\Sigma$.
\fi
\iffalse
Assume the corresponding Kalman filter with sensor output $Y_k$ has a steady state estimation error covariance matrix $P'$ and a Kalman gain matrix $K'$ that satisfy
%Similarly, the Kalman filter parameter $\mathbf{K}'$ is 
\begin{center}
$K'=P'C^T\Sigma^T(\Sigma CP'C^T\Sigma^T+\Sigma R)^{-1}.$
\end{center}
\fi
%%%%%%%%%%%%%%%%%%
%%%%%%%%%%%%%%%%%%
%%%%%%%%%%%%%%%%%%%%%%
%%%%%%%%%%%%%%%%%%%%%%
\subsection{When sensor and actuator packets are both injected}
\label{sec:actuator_inject}
We will derive the condition for a feasible coding matrix when the attacker can mount deception attacks to both sensor packets and actuator packets. 
\begin{theorem}
Given an attacked system model~\eqref{attackmodel}, assume that the attacker designs a sequence of stealthy sensor and actuator data injection $(y^a_k, u^a_k), k=0,1,\dots$, that $u^a_k$ is bounded and drives the estimation error to infinity $\lim_{k\to\infty}\|\Delta e_k\|_2 \to \infty$. If there exists an invertible matrix $\Sigma$ such that $y_k^a-\Sigma^{-1} y_k^a\neq 0$ for any $y_k^a$, 
%\begin{align}
%\end{align}
then after injecting $(y^a_k, u^a_k)$ the estimation residue change satisfies $\lim_{k\to\infty}\|\Delta z'_k\|_2 \to \infty$, by coding sensor outputs~\eqref{sig_y} with $\Sigma$. 
\end{theorem}

\begin{proof}
The dynamics of change of estimation error, residuals between the normal and compromised system is described as~\eqref{delta_z},
%\begin{align}
%\begin{split}
%&\Delta e_{k+1} = (A-KCA) \Delta e_{k} - Ky^{a}_{k+1}+(I-KCB)u_k^a,\\
%&\Delta z_{k+1}=CA\Delta e_k+ y^a_{k+1}+CBu^a_k,
%\end{split}
%\end{align}
where $y^a_k, u_k^a$ is the injected sequence to sensor and actuator packets, respectively. Since $\|\Delta z_{k+1}\|_2 \leq M$ for all $k=0, 1,\dots$, any pair of $(y^a_{k+1}, u^a_k)$ must satisfy 
\begin{align}
y^a_{k+1}=-CA\Delta e_k - CBu^a_k+ \epsilon_k,\  \|\epsilon_k\|_2 \leq M.
\label{yak}
\end{align}
For bounded $u^a_k$, the injection sequence satisfies that $\lim_{k\to\infty}\|y^a_k\|_2 \to \infty$ to make sure $\lim_{k\to\infty}\|\Delta e_k\|_2 \to \infty$. When coded sensor values are injected as~\eqref{sig_ya}, and the estimator decodes the value as
\begin{align*}
% \\\centerline{$
\tilde{y}'_k=\Sigma^{-1} Y'_k=Cx_k+v_k+\Sigma^{-1} y^a_k,%$}
\end{align*}
the coded system with the original design of Kalman filter is equivalent to be injected by a sequence of pair $(\Sigma^{-1}y^a_{k+1}, u^a_k)$. It is worth noting that the actuator data is not coded, and $u^a_k$ keeps the same for both the original and coded system.
The dynamics of the change of estimation error, residuals between the normal and compromised coded system are as following
%\footnotesize
\begin{align*}
\begin{split}
&\Delta e'_{k+1} = (A-KCA) \Delta e'_{k} - K\Sigma^{-1}y^{a}_{k+1}+(B-KCB)u_k^a,\\
&\Delta z'_{k+1}=CA\Delta e'_k+ \Sigma^{-1}y^a_{k+1}+C B u^a_k.
\end{split}
\end{align*}
%\normalsize     
Without loss of generality, we assume that $\Delta e_0 =0$, then 
%\footnotesize
\begin{align*}
\Delta e_k &=\sum_{j=1}^{k} (A-KCA)^{k-j} (- K y^{a}_{j}+(B-KCB)u_{j-1}^a),\\
\Delta e'_k &=\sum_{j=1}^{k} (A-KCA)^{k-j} (- K\Sigma^{-1}y^{a}_{j}+(B-KCB)u_{j-1}^a).
\end{align*}
%\normalsize
Plug in the expression of $\Delta e'_k$ in the equation of $\Delta z'_{k+1}$, with $CBu^a_k=-y^a_{k+1}-CA\Delta e_k + \epsilon_k$, we have
%$ \Delta z'_{k+1}$, we have$\Delta e_k$, 
%\footnotesize
\begin{align}
\begin{split}
\Delta z'_{k+1}=& CA\sum_{j=1}^{k} (A-KCA)^{k-j} (- K\Sigma^{-1}y^{a}_{j}\\
                      & +(B-KCB)u_{j-1}^a)+\Sigma^{-1}y^a_{k+1}+C B u^a_k\\
                %  =& CA\sum_{j=1}^{k} (A-KCA)^{k-j-1}K(I-\Sigma^{-1})y^a_j\\
                %   &+CA\Delta e_k -CA\Delta e_k +(\Sigma^{-1}-I) y^a_{k+1}+ \epsilon_k\\
                   =&CA\sum_{j=1}^{k} (A-KCA)^{k-j}K(I-\Sigma^{-1})y^a_j\\
                     &+(\Sigma^{-1}-I) y^a_{k+1}+ \epsilon_k.
                                         %(-K\Sigma^{-1} (-CA\Delta e_{j-1} - Cu^a_{j-1}+ \epsilon)+)
\end{split}
\label{z_uy}                                        
\end{align}
%\normalsize
Hence, for $\Sigma \neq I$,  $\lim_{k\to\infty}\|y^a_{k+1}\|_2 \to \infty$, we have $\lim_{k\to\infty}\|\Delta z'_k\|_2 \to \infty$ for $\Delta z'_k$ defined in~\eqref{z_uy}.
%the injection sequence designed by the attacker cannot keep 
%$\|\Delta z'_k\| \leq M$ for all  $k=1,2, \dots$ for all $\Delta z'_k$ defined in~\eqref{z_uy}, with $(y_k^a,u_k^a)$ designed to make the original system estimation error $\lim_{k\to\infty}\|\Delta e_k\|_2 \to \infty$. Then the injection cannot keep stealthy. 
\end{proof}

%\subsection{When Sensor Networks Have Structural Constraints}

\section{Algorithm to Compute A Coding Matrix}
\label{algorithm}
In this section we propose an algorithm to compute a set of feasible coding matrices for the case there exists a sequence of sensor data injections to cause unbounded state estimation error, i.e., the system has unstable eigenvectors of $A$.

The coded sensor values should increase the difference between estimation residue of the normal and attacked system -- $\|\Delta z'_k\|_2$ as $k \to \infty$, which is equivalent to keep $\|Cv-\Sigma^{-1} Cv\|_2$ or $\|C\tilde{v}- \Sigma^{-1}C\tilde{v}\|_2$ for multiple unstable eigenvectors nonzero, by the proof of Theorem~\ref{code} and Lemma~\ref{code_lemma}. The system satisfies that $(A, C)$ is detectable, then with an invertible coding matrix $\Sigma$ and the decoded sensor value $\tilde{y}'_k$ defined in~\eqref{y_decode}, $\tilde{y}'_k=y_k$ when $y^a_k=0$. Hence, the state estimator still converges to the true state without attacks and the coding scheme does not sacrifice the performance of state estimator.

For multiple unstable eigenvectors, when we do not know the exact linear combination result of $\tilde{v}$ applied by the attacker to design the injection sequence, we can not guarantee that $\Sigma$ works for the exact injected sequence $y^a_k$ by finding a feasible coding matrix with respect to a specific vector $v$. According to Theorem~\ref{code} and Lemma~\ref{code_lemma}, the coding matrix should work for any possible injection sequence $y^a_k$ designed based on unstable eigenvectors of the system matrix $A$. Hence, we consider to find a coding matrix based on the concept of a rotation matrix without specific knowledge about the value of injected data to sensors.
%Given the system parameter matrices $C$, $A$, the objective of~\eqref{opt_code} is approximately to find an invertible, bounded $\Sigma$ that maximizes $\|\Sigma Cv- Cv\|_2$.  
 
%Considering the direction change of vector $Cv$ after transformed by $\Sigma Cv$, the optimal direction to maximize the difference between $\Sigma Cv$ and  $Cv$ is the orthogonal direction. Thus, we have the following Lemma~\ref{orthogonal}. 
%$(Cv)'\Sigma Cv=0,\Sigma \neq 0$ provides the optimal direction. 
%%%%%%%%%%%%%%%%%%%
%%%%%%%%%%%%%%%%%%%
%%%%%%%%%%%%%%%%%
%%%%%%%%%%%%%%%%%%%%%%
%The following lemma presents the constraints of an optimal $\Sigma$ from the geometric insight. 
%For a steady state Kalman filter, $(A-KCA)$ is a stable matrix since 
%%%%%%%%%%%%%%
%%%%%%%%%%%%

%\subsection{Rotation matrix}

\begin{definition}
A Givens rotation is a $n\times n$ rotation matrix, with $1$'s on the diagonal, $0$'s elsewhere, except the intersections of the $i$th and $j$th rows and columns corresponding to a rotation in the $(i,j)$ plane in $n$ dimensions. It takes the following form
\begin{align}
G(i,\ j,\ \theta)=\begin{bmatrix}1&\cdots&0&\cdots&0&\cdots&0\\ \vdots &\ddots &\vdots& &\vdots& &\vdots \\ 
0&\cdots & c &\cdots & -s & \cdots &0 \\ \vdots & &\vdots&\ddots&\vdots& &\vdots \\
0&\cdots&s&\cdots&c&\cdots&0 \\ \vdots& &\vdots& &\vdots&\ddots&\vdots \\ 0&\cdots&0&\cdots&0&\cdots&1
\end{bmatrix},
\label{grotate}
\end{align}
where $c=\cos \theta$, $s=\sin \theta $. 
\end{definition}

The product $G(i,j,\theta) x$ represents a counterclockwise rotation of the vector $x\in \mathbb{R}^p$ in the $i,j$ plane of $\theta$ radians. Hence, only the $i$-th and $j$-th elements of $x$ will be changed.  Given system model~\eqref{system}, there are multiple ways to choose a rotation matrix as a coding matrix in general. If a rotation matrix can guarantee that the direction of any possible stealthy injection is changed, it must rotate all nonzero elements in the vector space
\begin{align}
span(Cv_1, \dots, Cv_u)
\label{inject_space}
\end{align}

The following algorithm provides a design process of a rotation matrix given system matrix $A$.
\begin{algorithm}
\caption{\textbf{: Compute a feasible coding matrix $\Sigma$}}
%\begin{algorithmic}
\textbf{Input}: System model parameters $A,C$, unstable eigenvalues and eigenvectors $\lambda_i, v_i,\  i=1,\dots, u$ of $A$.
\\\textbf{Initialization}: Calculate vectors $Cv_i \in \mathbb{R}^p$ for all unstable eigenvectors $i=1,\dots, u$. Construct the standard basis $(e_{p_1},\ e_{p_2}, \dots, e_{p_l})$, $e_{p_j} \in \mathbb{R}^p$ for the vector space defined as~\eqref{inject_space}, where $1\leq p_1 < p_2 \dots < p_l \leq p$, and $e_{p_j}$ is a vector with the $p_j$-th element as $1$ and all the other elements as $0$. Define rotation step as $r=1$, uncovered unstable dimension set as $S=\{p_1, p_2,\dots, p_l\}$.
\\\textbf{Iteration}:  When $S\neq \phi $\\
If more than two elements are left in the set $S$: randomly picking up a rotation radian $\theta \in (0, \frac{\pi}{2}]$, rotation dimension $p_i, p_j \in S$, let $S=S\setminus \{p_i, p_j\}$;\\
Else: randomly picking  a rotation radian $\theta \in (0, \frac{\pi}{2}]$ with uniform distribution, rotation dimension $p_i\in S$, $ p_j \in \{1,\dots,p \}$ and $p_j \neq p_i$, let $S=S\setminus \{p_i\}$. \\
Get the rotation matrix $G_r=G(p_i,p_j,\theta)$ as defined in~\eqref{grotate}.
Let $r=r+1$.
\\\textbf{Return}: A feasible transform matrix $\Sigma = G_1 G_2 \dots G_r$. 
\label{calculate_g}
\end{algorithm}

The existence condition of a feasible coding matrix designed as a rotation matrix is then explained in the following lemma. 
\begin{lemma}
When the dimension of matrix $C$ of the system~\eqref{system} satisfies that $p \geq 2$, there always exists a feasible givens rotation matrix $\Sigma$ that satisfies the condition of Theorem~\ref{code} or Lemma~\ref{code_lemma} for the system. 
\label{exist}
\end{lemma}
\begin{proof}
According to the definition of a Givens matrix~\eqref{grotate} and the process of calculating a feasible rotation matrix,  when $p \geq 2$, we apply Algorithm~\ref{calculate_g}. Since every rotation has an angle $\theta \in (0,\ \frac{\pi}{2}]$ and there are no two rotations in the same plane, vector $\Sigma Cv$ is not in the same direction with $Cv$. Hence, Algorithm~\ref{calculate_g} provides a feasible rotation matrix.
\end{proof}

The coding scheme proposed in this work is a low cost approach from computation perspective. Specifically, the proposed coding scheme requires only $O(n^3+p^3)$ multiplications and additions, where n and p denote the number of plant states and sensors respectively. As we clarify now in the new version of the manuscript (in Section IV), this is significantly lower than the computation cost for even basic encryption and coding schemes that involve computation of highly complex non-linear primitives~\cite{foundation_encrypt,encrypt_sensor,correct_code1977}.

The coding scheme proposed in this work is also a low cost approach from communication perspective. The coding scheme proposed in this work does not require additional bits for each plaintext message of the sensor measurements, while an encryption method introduces communication overhead for each sensor message transmitted in the communication channel~\cite{encrypt_key}. The sensor outputs coding approaches proposed in this work aim to change the value transmitted over the communication channel instead of correcting errors on bit level compared with error-correcting additional coding bits~\cite{correct_code1977}. Hence, the communication overhead of the proposed scheme in this work is relatively low.

\begin{remark}
The rotation matrix $\Sigma$ calculated by Algorithm~\ref{calculate_g} is a sparse matrix in general, since a rotation matrix has many $0$ elements, and Algorithm~\ref{calculate_g} is a polynomial heuristic algorithm. This means the coding process is computationally efficient.
\end{remark}
%Moreover, , that the computing overhead of finding a feasible coding matrix is
%For the case when both 

For systems with structural constraints, two potential schemes can be considered. One is that the structure of $\Sigma$ is also limited and we design a coding matrix $\Sigma$ with an additional constraint that some components $\Sigma$ must be $0$ because of the sparsity of the sensors the system equipped with. Another scheme is distributed coding that multiple coding matrices are applied for the whole system. This is a revenue for future work.

\section{Time-Varying Coding Scheme When the Attacker Estimates the Coding Matrix}
\label{Sig_t}
%\begin{remark}
The coding scheme in this work is effective for the cases that sensor values are not manipulated by the attacker before they are coded by matrix $\Sigma$. 
%\footnote{emphasize this in introduction and abstract, and find out examples (like GPS spoofing attacks) applying to the communication channel.}
We also assume that the attacker does not know when the system starts to apply $\Sigma$ for transforming sensor output values, and aims to inject a stealthy sequence $y^a_k$ to the sensor communication channel with respect to the original system. If the attacker is powerful enough to update the system model and acquire the knowledge of the coding design after some time steps, the system should constantly apply a time-varying coding scheme, and the time length for updating the coding matrix depends on the learning ability of the attacker and detecting requirements of the system.

%can apply a random coding technique. 
%When the system randomly switches among a set of feasible coding matrices---for instance, according to a mixed strategy game between the system and the attacker---
Each time the system updates the coding matrix, it will cost the attacker some time to figure out the transformed sensor outputs values. Since it is sufficiently fast to compute a feasible transform based on the algorithm, the system can even generate new coding matrices during the running process. Before the attacker learns $\Sigma$ or the coded observer parameter $\Sigma C$, the false data injection sequence is not stealthy for the coded system. We assume that the attacker cannot directly acquire the coding matrix during its communication process, similar as the secrecy requirement of a key for encryption sensor nodes~\cite{encrypt_key, encrypt_sn}. We assume that the sensors and controller are synchronized, which is a standard assumption in safety-critical control systems. Thus, with the same notion of time, both sensors and the controller can use the same random generator to (re)generate the coding matrix or exploit some of the existing schemes for secret key distribution. In addition, they will be able to synchronously switch from using one matrix to the newly created/obtained ones. Various protocols of key distributions have been proposed according to the properties of the systems~\cite{encrypt_key, key_protocol}.
%the attacker only has a probability to stay undetected for some time. 
\subsection{The time length an attacker needs to learn $\Sigma$}
To learn the matrix $\Sigma$ that distributed secured between sensors and the controller/estimator/detector, we assume that the attacker is able to eavesdrop the sensor outputs and actuator inputs via the communication channel for estimating $\Sigma$, instead of directly capturing the matrix $\Sigma$. Since $y^a_k$ is designed by the attacker, the sensor information received by the attacker is then the true sensor measurements under the coding scheme $Y_k=\Sigma y_k$. System dynamics from the perspective of an attacker are %the knowledge space of the attacker is $$
%\footnotesize
\begin{align}
\begin{split}
x_k=&A^k x_0+\sum_{j=0}^{k-1} A^{k-j-1}(Bu_j+w_j),\\
Y_k=&\tilde{C} A^k x_0+\sum_{j=0}^{k-1} \tilde{C}A^{k-j-1}(Bu_j+w_j) +y^a_k+v_k,
%x_k=&
%Y_k=&Y_k'-y_k^a=\Sigma (Cx_k+v_k)
\end{split}
\label{at_ss}
\end{align}
%\normalsize
where $\tilde{C}=\Sigma C$.  
When the attacker does not have any knowledge about the structure of the coding matrix $\Sigma \in \mathbb{R}^{p \times p}$, there are $p^2$ variables for estimating $\Sigma$. Meanwhile, in general initial state $x_0$ can only be acquired via estimation, and there are $n$ variables additionally in~\eqref{at_ss}. 
Without loss of generality, we initialize $k=0$ as the time that attacker starts to observe the system's sensor outputs and actuator inputs to update the knowledge of the system coding scheme. It is worth noting that for designing a sequence of stealthy injection data, the attacker needs to know the model of the system, including the estimator and statistics detector, while the values of sensor outputs or actuator inputs are not necessary for the attacker. When the attacker starts to record sensor and actuator communicational packets at an arbitrary time $k$, the corresponding system state $x_0$ can not be directly retrieved by the attacker. Hence, $\Sigma \in \mathbb{R}^{p \times p}$ and $x_0 \in \mathbb{R}^n$ are variables to be estimated.

We examine a simpler case to estimate the the coding matrix first---how many steps of sensor values the attacker need to measure for the following noise-free LTI system 
\begin{align}
\begin{split}
x_{k+1}=&Ax_k+Bu_k,\quad
\bar{y}_k=Cx_k.
\end{split}
\label{nf_sys}
\end{align}
The sensor outputs coded by $\Sigma$ at time $k$ are
\begin{align}
\bar{Y}_k=\tilde{C} A^k x_0+\sum_{j=0}^{k-1} \tilde{C}A^{k-j-1}Bu_j,\quad \bar{Y}_k \in \mathbb{R}^{p}.
\end{align}
We define the attacker's observation $Y_{\Sigma,N}$ during time $k=0,1,\dots, N$ when the system applies $\Sigma$, and the corresponding noise-free measurements $\bar{Y}_{\Sigma, N}$ as
\begin{align*}
%\begin{split}
%\centerline{$
Y_{\Sigma, N}=\begin{bmatrix}Y_0|Y_1|\cdots|Y_N\end{bmatrix},\quad
\bar{Y}_{\Sigma, N}=\begin{bmatrix}\bar{Y}_0|\bar{Y}_1|\dots|\bar{Y}_N \end{bmatrix}.%$}
%\end{split}
%\label{YN}
\end{align*}
The observed sensor values from the perspective of the attacker are bilinear equations with respect to $\Sigma$ and $x_0$. Consider the noise-free dynamics of sensor measurements as the following
%\footnotesize
\begin{align}
%\begin{split}
%\centerline{$
\bar{Y}_{\Sigma,N}=\Sigma C \begin{bmatrix} x_0&\cdots&A^N x_0+\sum_{j=0}^{N-1} A^{N-j-1}Bu_j\end{bmatrix},%$}
%&Ax_0+Bu_0&
%&=\Sigma C\left (\begin{bmatrix}I&A&\cdots&A^k\end{bmatrix}diag(x_0)\\
%&+\begin{bmatrix}0&B&AB&\cdots&A^{k-1}B\end{bmatrix}\begin{bmatrix}0&\cdots \\ 0&u_0&u_1&\cdots&u_{k-1}\\ 0&0&u_0&\cdots &u_{k-2}\\\vdots\\0&\cdots&\cdots&\cdots &u_0\end{bmatrix} \right)
%\end{split}
\label{YN}
\end{align}
%\normalsize
where $[\bar{Y}_1,\dots, \bar{Y}_N],\quad [u_0,\dots, u_N]$ is eavesdropped by the attacker and $(A,B,C)$ is within the knowledge space of the attacker. To write the above equation as a standard form of bilinear equations regarding to vectors, we denote the coding matrix $\Sigma$ as
%\footnotesize
\begin{align*}
%\\\centerline{$
\Sigma=\begin{bmatrix} \Sigma_{11} & \Sigma_{12}&\cdots &\Sigma_{1p}\\ \vdots &\vdots &\ddots & \vdots \\ \Sigma_{p1} & \Sigma_{p2} & \cdots & \Sigma_{pp} \end{bmatrix}=\begin{bmatrix}\sigma_1\\ \vdots \\ \sigma_{p} \end{bmatrix},%$} %\in \mathbb{R}^{p^2}
\end{align*}
%\normalsize
where $\sigma_i \in \mathbb{R}^{1 \times p},\  i \in\{1,\dots, p\}$ is the $i$-th row of matrix $\Sigma$.
%$\Sigma \in \mathbb{R}^{p\times p}$, $x_0 \in \mathbb{R}^n$ is to be estimated by the attacker.  
We also vectorize $\bar{Y}_{\Sigma, N} \in \mathbb{R}^{p \times (N+1)}$ $(Y_{\Sigma, N})\in \mathbb{R}^{p \times (N+1)}$ as $\bar{d} \in \mathbb{R}^{p(N+1)}$ 
$(d \in \mathbb{R}^{p(N+1)})$
%\footnotesize
\begin{align}
\begin{split}
&\text{vec}(\bar{Y}_{\Sigma, N}) =\begin{bmatrix}[\bar{Y}_0]_1\\ \vdots \\ [\bar{Y}_0]_p \\ \vdots \\ [\bar{Y}_N]_p\end{bmatrix}
=\begin{bmatrix}\bar{d}_1 \\ \vdots \\ \bar{d}_p \\ \vdots \\ \bar{d}_{p(N+1)}\end{bmatrix} \in \mathbb{R}^{p(N+1)},
\end{split}
\label{vecY}
\end{align} 
%\normalsize
where $[\cdot]_j$ means the $j$-the element of a vector, and $\bar{d}_i \in \mathbb{R}$. 
Then equation~\eqref{YN} can be written as the following $p(N+1)$ equations 
%\footnotesize
\begin{align}
\begin{split}
&\sigma_i C x_0=[\bar{Y}_0]_i=\bar{d}_i,\\
&\sigma_i (CA^k) x_0 + \sigma_i \left(C\sum_{j=0}^{k-1} A^{k-j-1}Bu_j\right) = [\bar{Y}_k]_i=\bar{d}_{pk+i}, \\
%& i=1,\dots, p,\ k=1,\dots, N
\end{split}
\label{bilinear}
\end{align}
%\normalsize
In particular, define coefficient matrices 
\begin{align*}
\begin{split}
T_0=C,\  T_{k}= CA^k, 
S_0=0,\  S_{k}=C\sum_{j=0}^{k-1} A^{k-j-1}Bu_j, 
\end{split}
%\label{ts}
\end{align*}
for $k=1,\dots, N$.  
For the case of a noise-free system, the attacker is possible to solve the bilinear problem~\eqref{bilinear} only after observing enough time steps of $\bar{Y}_k$.
\begin{remark}
By the property of bilinear equations~\cite{bilinear}, the attacker needs at least $N \geq \text{max} \{n, p\}-1$ measurements of sensor and actuator values to calculate the exact coding matrix $\Sigma$ and true initial state $x_0$ when there is no noise.
\end{remark}

With noises in practical, we have
\begin{align}
\begin{split}
&\sigma_i C x_0 + [v_0]_i=[Y_0]_i=d_i,\\
&\sigma_i (CA^k) x_0 + \sigma_i \left(C\sum_{j=0}^{k-1} A^{k-j-1}(Bu_j+ w_j)\right)+[v_k]_i \\
=& [Y_k]_i=d_{pk+i},\quad i=1,\dots, p,\ k=1,\dots, N,
\end{split}
\label{equal_bilinear}
\end{align}
where $[Y_k]_i$ and $d_i$ are defined similar as $[\bar{Y}_k]_i$ an $\bar{d}_i $  in vectorization~\eqref{vecY}. Under the assumption that both $w_k$, $v_k$ are i.i.d. Gaussian noise, for any $k$, their expectations satisfy
 \begin{align*}
& \mathbb{E}\left [ \sigma_i (C\sum_{j=0}^{k-1} A^{k-j-1} w_j)+[v_k]_i \right]=0.
 \end{align*}
Then the noise-free and noisy sensor values satisfy that $ \mathbb{E} Y_{\Sigma, N} = \bar{Y}_{\Sigma, N}.$
 
 Hence, when the attacker observes noisy sensor outputs $Y_{\Sigma, N}$, the objective of retrieving the coding matrix $\Sigma$ without the knowledge of $x_0$ is equivalent to finding $\sigma_1,\dots, \sigma_p, x_0$ that fit for the noise-free equation set~\eqref{bilinear}. With even Gaussian noise, it becomes difficult to numerically find an exact solution of the true coding matrix, and the problem is then to minimize the total error between the left and right sides of the equations. The problem of estimating $\Sigma, x_0$ is described as Problem~\ref{estlemma}.
   
 \begin{problem}
 The problem of estimating $\Sigma, x_0$ in the minimum mean square error perspective is defined as
 the following bilinear programming problem
 \begin{align}
 \begin{split}
 \underset{\sigma_1,\dots, \sigma_p, x_0}{\text{minimize}}\quad &\sum_{k=0}^{N} \sum_{i=1}^{p}\left\|\sigma_i T_k x_0 + \sigma_i S_k- d_{pk+i} \right\|_2\\%\|\Sigma CAK'\Delta y\|_2 \\
\text{subject  to}\quad  %&\|\Sigma\|_2 \leq \gamma,\\
                                        & \text{rank}(\Sigma = \begin{bmatrix}\sigma_1 \\ \vdots \\ \sigma_p\end{bmatrix})=p.  %& \Sigma\ \text{invertible}.
 \end{split}
 \label{est_bil}
 \end{align}
\label{estlemma}
\vspace{-8pt}
 \end{problem}
When there exists an invertible matrix $\Sigma$ that satisfies the equations defined in~\eqref{equal_bilinear}, the above bilinear optimization problem~\eqref{est_bil} has an optimal cost $0$. However, the optimal solution $\Sigma^*$ does not need to be the true coding matrix $\Sigma$, since there is noise and the objective function of problem~\eqref{est_bil} does not include noise of each time step.
%Since there exists noise and the true $\Sigma$ satisfies equations~\eqref{bilinear} instead of~\eqref{est_bil}. When the optimal cost of problem~\eqref{est_bil} is $0$, the optimal solution of $\Sigma^*$ is not the true coding matrix $\Sigma$ applied by the system,
 
The rank constraint of problem~\eqref{est_bil} is non-convex, and in practice the attacker does not know how many measurements return the best estimation before calculating $\Sigma^*$ given all existing measurements. Hence, we design the following heuristic algorithm for the attacker, which ignores the rank constraint first, and checks whether $\Sigma$ is full rank every step till a feasible solution is reached. 
%\iffalse
\begin{algorithm}
\caption{Algorithm of estimating $\Sigma$ for the attacker}
\textbf{Inputs:} {System's parameter $(A,B,C)$, design of Kalman Filter $K$, the threshold $\alpha$ of $\chi^2$ detector, algorithm stopping condition--estimation error $\epsilon$.}\\
\textbf{Initialization:} {Initialize the value of estimation error $Er > \epsilon$, and the estimation of coding matrix $\hat{\Sigma}$ as a  n identical matrix.}\\
%\While{$Er > \epsilon$}{
\textbf{While $Er > \epsilon$ or $\hat{\Sigma}=I$.}\\
1). Read one new sensor and actuator observation, and update parameters of problem~\eqref{est_bil}; \\
2). Solve problem~\eqref{est_bil}. If the optimal solutions $\sigma^*_1, \dots, \sigma^*_p$ satisfy the full rank constraints, let $Er$ be the value of the optimal cost, and $\hat{\Sigma}= \Sigma^*=\begin{bmatrix}(\sigma^*_1)^T  \dots  (\sigma^*_p)^T\end{bmatrix}^T$.\\
\textbf{Return:} {Estimation result of $\hat{\Sigma}$.} 
\label{attacker_alg}
\end{algorithm}
\begin{remark}
It is worth noting that a bilinear equation usually has multiple solutions, and Algorithm~\ref{attacker_alg} returns different optimal solutions $\hat{\Sigma}$ under different sensor and actuator measurements time $N$. Under this situation, it is not clear for the attacker to decide how many time steps to measure and which optimal solution to choose, even when the optimal cost of problem~\eqref{est_bil} is $0$. Even for a simple two dimensional system $A$, multiple solutions exist and do not converge to one estimation after $20$ steps of measurements, which we will show in simulation.
\end{remark}

To summarize, there are two main challenges for the attacker to estimate the true coding matrix, the first one is because multiple solutions exist for bilinear equations or bilinear optimization problems. The second one comes from the noise in the communication channel, that even the attacker find a feasible solution to the bilinear equation set~\eqref{equal_bilinear}, it is only an unbiased estimation instead of the true coding matrix. 

\subsection{When the estimated $\hat{\Sigma} \neq \Sigma$}
After the attacker estimates a coding scheme $\hat{\Sigma}$ and considers it as the true coding matrix the system is applying, the easiest way to keep stealthy is to inject $\hat{\Sigma} y_k^a,$ where $y_k^a$ is a stealthy data injection designed for the original system without coding. However, as discussed above, when there exists noise, the attacker is not able to calculate the exact coding matrix the system is applying. When $\hat{\Sigma}\neq \Sigma$, the injection sequence $\hat{\Sigma} y_k^a$ can only extend the time length before detected and cannot pass the detector. Then the system needs to evaluate how long the attacker needs to measure the sensor outputs and how long the attacker can stay stealthy by applying a new injection sequence, in order to decide the time length of changing the coding matrix.

\begin{definition}
An estimated coding matrix $\hat{\Sigma}$ calculated by the attacker is called a feasible estimation of $\Sigma$ that keeps the attacker stealthy for time $k=0, \dots,T$ while causing $e$ error, if and only if  for all sequence of injections $\hat{\Sigma}y_k^a$ designed by the attacker according to the estimated coding matrix $\hat{\Sigma}$,  the dynamics of $\Delta e'_k, \Delta z'_k$ satisfy that
\begin{align}
\begin{split}
&\Delta e'_{k+1}=(A-KCA )\Delta e_{k} - K\Sigma^{-1} \hat{\Sigma} y^{a}_{k+1},\\
&\Delta z'_{k+1}=CA\Delta e_k+ \Sigma^{-1} \hat{\Sigma} y^a_{k+1}, \\
&\underset{k=1,\dots, T}{\text{max}}\|\Delta e_k\|_2 \geq e, \  \underset{k=1,\dots, T}{\text{max}}\|\Delta z_k \|_2 \leq M,
\end{split}
\end{align} 
where $\Sigma$ is the true coding matrix the system is applying.
\end{definition} 
 
 Define the time length of keeping stealthy with injection sequence $y^a_k, k=1,\dots, T_s,$ for a system~\eqref{system} as
 \begin{align}
 T_s (y^a_k) =\underset{k}{\text{inf}}\{k: \|\Delta z_k (y^a_k) \|_2 > M\}.
 \end{align}
 The attacker increases the time length of keeping stealthy when $T_s(\hat{\Sigma}y^a_k) > T_s (y^a_k)$.  However, the attacker does not have a guarantee about $T_s(\hat{\Sigma}y^a_k)$ without the knowledge of the true coding matrix, since $\Delta z_k(\hat{\Sigma}y^a_k)$ is affected by both $\Sigma$ and $\hat{\Sigma}$. There exists a trade-off between the time $N$ the attacker takes to measure sensor and actuator values to estimate a better $\hat{\Sigma}$ and the time the attacker starts to apply a new injection sequence $\hat{\Sigma} y^a_k$. If the measuring time $N$ is large, it is possible that the system already triggers the alarm before the attacker successfully recovers the coding scheme. If the attacker does not have enough measurements for a good estimation and then applies the estimated $\hat{\Sigma}$ to design a new injection sequence $\hat{\Sigma}y^a_k$, $T_s(\hat{\Sigma}y^a_k)$ will not be much larger than $T_s (y^a_k)$ and the malicious behavior will still be detected quickly by the system. 

It is worth noting that the system can not decide whether $\Sigma$ is easy to be estimated by the attacker by only checking $\|\Delta z_k (y^a_k) \|_2, k=1,\dots, T_s(\hat{\Sigma}y^a_k)$. When $\|\Delta z_k (y^a_k) \|_2$ stays in a small range for a long time and  $T_s(\hat{\Sigma}y^a_k)$ is large, the reason may be the original injection sequence $y^a_k$ also has a large time length of keeping stealthy $T_s(y^a_k)$. Define the stealthy time increasing proportion for an estimated $\hat{\Sigma}(N)$  calculated after measuring time $N$ as
\begin{align}
%\\\centerline{$
\alpha(N_{\Sigma})=\frac{T_s( \hat{\Sigma}(N)y^a_k)-T_s(y^a_k)}{T_s (y^a_k)},%$}
\end{align} 
where $\hat{\Sigma}(N)$ is estimated from $N$ steps measurements of sensor and actuator values. As we will show in Section~\ref{sec:simulation}, $\alpha(N_{\Sigma})$ increases with an increasing $N_{\Sigma}$ for a fixed $\Sigma$ in general. When the attacker is able to estimate the coding matrix and inject $(\hat{\Sigma}(N)y^a_k)$ to stay stealthy for a longer time, the system needs to apply a new coding matrix before the attacker has enough measurements to estimate an $\hat{\Sigma}(N)$ that reaches the threshold $\tilde{\alpha}(N_{\Sigma})$ of the increasing time proportion $\alpha(N_{\Sigma})$. From the perspective of the system, a heuristic way to decide the time length $N_{\Sigma}$ of changing $\Sigma$ is as Algorithm~\ref{change_time}.

\begin{algorithm}
\caption{Heuristic Algorithm for choosing $N_{\Sigma}$}
\textbf{Inputs:}{Coded system's parameter $(A,B,C,K, \Sigma)$, $\chi^2$ detector threshold $\alpha$, time step $t_s$ for increasing $N_{\Sigma}$, threshold proportion $\tilde{\alpha}(N_{\Sigma})$.}\\
\textbf{Initialization:}{Initialize the value of $\hat{\Sigma}$ as an identical matrix, let $N_{\Sigma}=0$, calculate $\alpha(N_{\Sigma})$}\\
%\While{$Er > \epsilon$}{
\textbf{While $\alpha(N_{\Sigma})<\tilde{\alpha}(N_{\Sigma})$}\\
1). Estimate $\hat{\Sigma}$ with $t_s$ steps of new sensor and actuator values, and update $\alpha(N_{\Sigma})$. \\
2). Let $N_{\Sigma}=N_{\Sigma}+t_s$, and save sensor and actuator values for next iteration.\\
%2). Solve problem~\eqref{est_bil}. If the optimal solutions $\sigma_1, \dots, \sigma_p$ satisfy the full rank constraints, let $\hat{\Sigma}= \Sigma^*$.\\
%}
%\If{Optimal solution of }{}\\
%\eIf{}{}{}
\textbf{Return:} {Measurement time length $N_{\Sigma}$ for estimating $\Sigma$.} 
\label{change_time}
\end{algorithm}

\section{Illustrative Examples}
\label{sec:simulation}
\subsection{Coding scheme detect stealthy data injection}
We show the effects of coding sensor outputs by examples of two-dimensional LTI systems. 
\iffalse
Figures include comparisons of the change of estimation residues $\Delta z_k$ (for system $(A,B,C,K)$) and $\Delta z'_k$ (for system $(A,B,\Sigma C, K')$).
By comparing the change of estimation error $\Delta e_k$ and $\Delta e'_k$, we show that estimation error of a coded system does not necessarily increase faster than the original system.
\fi 
Consider a detectable 2-dimensional linear system with parameters: %~\cite{false_injection}:
%\footnotesize
%\begin{align*}
\begin{center}
$\mathbf{A}=\begin{bmatrix}0.8&0\\0.5&1\end{bmatrix}, \mathbf{B}=\begin{bmatrix}1\\0.5\end{bmatrix},
\mathbf{C}=\begin{bmatrix}2&0.5\\0&1\end{bmatrix}, \mathbf{D}=0,$
%&\mathbf{K}=\begin{bmatrix}0.3326&0.0397\\0.1774&0.0980\end{bmatrix}.
%\mathbf{L}=\begin{bmatrix}-1.0285&-0.4345\end{bmatrix},\\
%&\Gamma=\begin{bmatrix}1&0\\0&1\end{bmatrix}
%\end{align*}
%\normalsize
\end{center}
where $A$ has an unstable eigenvalue $\lambda=1$ and eigenvector $v=[0\ 1]^T$. 
%This system is detectable according to Popov-Belevitch-Hautus (PBH) rule, since
%\begin{align*}
%$\text{rank} \begin{bmatrix}C\\I-A \end{bmatrix}=2.$
%\end{align*}
One stealth attack sequence is: 
%\begin{center}
%\footnotesize
%satisfying the stealth condition is:
%\begin{align*}
%\begin{split}
$y_0^a =[0.0588\ 0.0588 ]^T$, $y^a_1=[0.1286\  -0.9706]^T$,  
$y_k=y_{k-2}^a-y_0^a, k \geq 2$.
%&\normalsize
%\end{split}
%\end{align*}
%\end{center}
%By Algorithm~\ref{iteration}, we get a feasible coding matrix $\Sigma_1$. Note that Theorem~\ref{code} does not require $\Sigma$ to be an SDP matrix. 
Multiple solutions of feasible coding matrices that satisfy Theorem~\ref{code} exist in general. For instance, for the above system, $\Sigma_1=\begin{bmatrix}2&-0.5\\-0.5&1\end{bmatrix}$ and $\Sigma_2=\begin{bmatrix}1&-1\\2&0\end{bmatrix}$ are both feasible. 
%rank \begin{bmatrix}\Sigma C\\I-A \end{bmatrix}=2. $$
%$$K'=\begin{bmatrix}1&2\\0&1\end{bmatrix}.$$

Figure~\ref{residue_1} shows the comparison result of $\|\Delta z_k\|_2, \|\Delta z'_k\|_2$ when there is injection attacks for the original and coded systems, and $\|\Delta z'_k\|_2$ increases with time $k$ after coded by $\Sigma_1$, while without coding $\|\Delta z_k\|_2$ is bounded. %for a randomly generate upper-triangular, invertible 
%$$\Sigma=\begin{bmatrix}1&2\\0&1\end{bmatrix}.$$
\begin{figure}[b!]
\vspace{-5pt}
\centering
\includegraphics [width=0.38\textwidth]{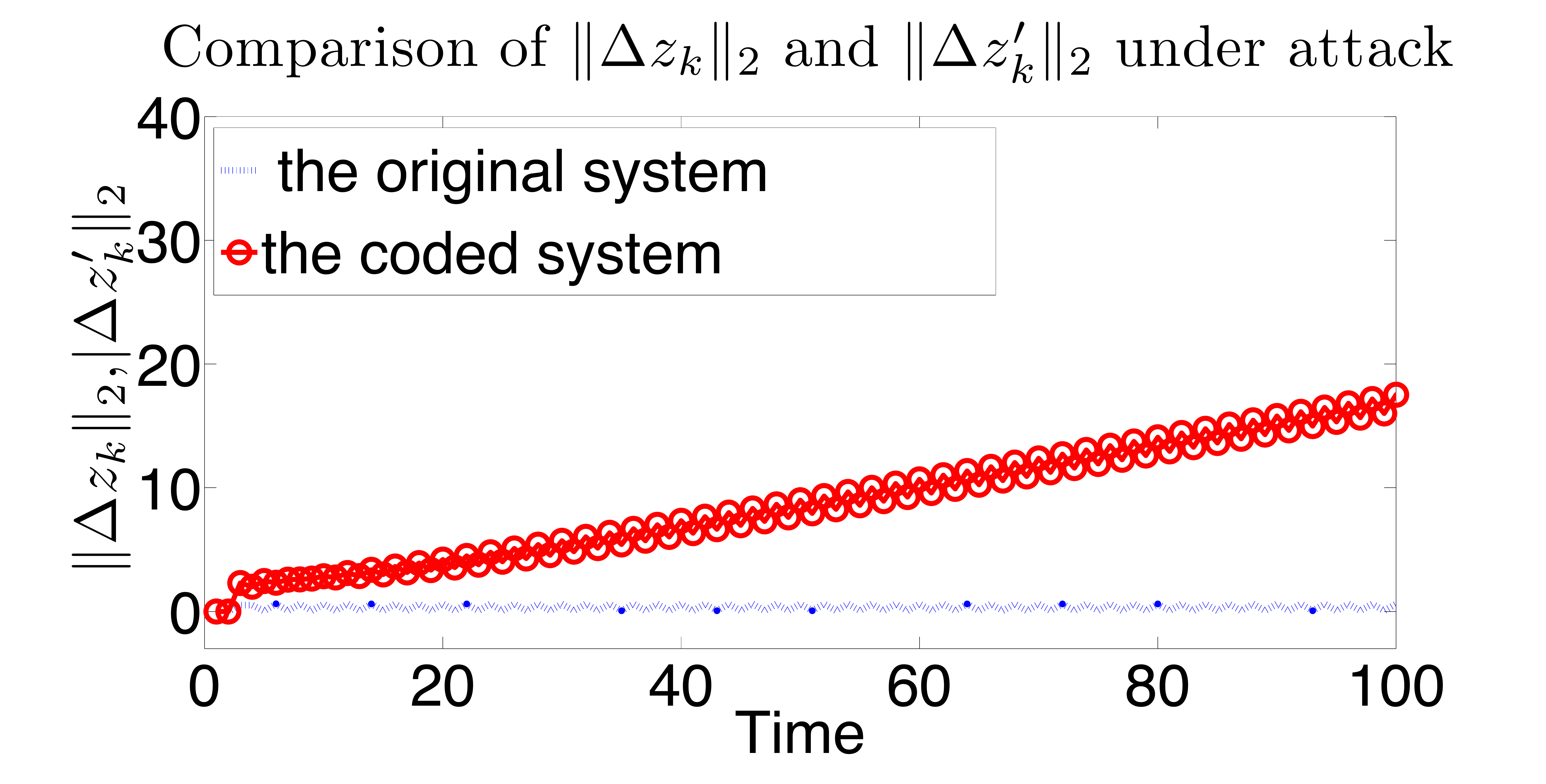}
\vspace{-10pt}
\caption{Comparison of $\|\Delta z_k\|_2$ of the original system and $\|\Delta z'_k\|_2$ of the coded system with $\Sigma_1$.}%between the original and coded systems.}
%For the coded sensor output, $\Delta z'_k$ increases with time $k$, while the original system $\Delta z_k$ stays in a bounded range.} 
\label{residue_1}
\end{figure}
\begin{figure}[t!]
%\vspace{-5pt}
\centering
\includegraphics [width=0.38\textwidth]{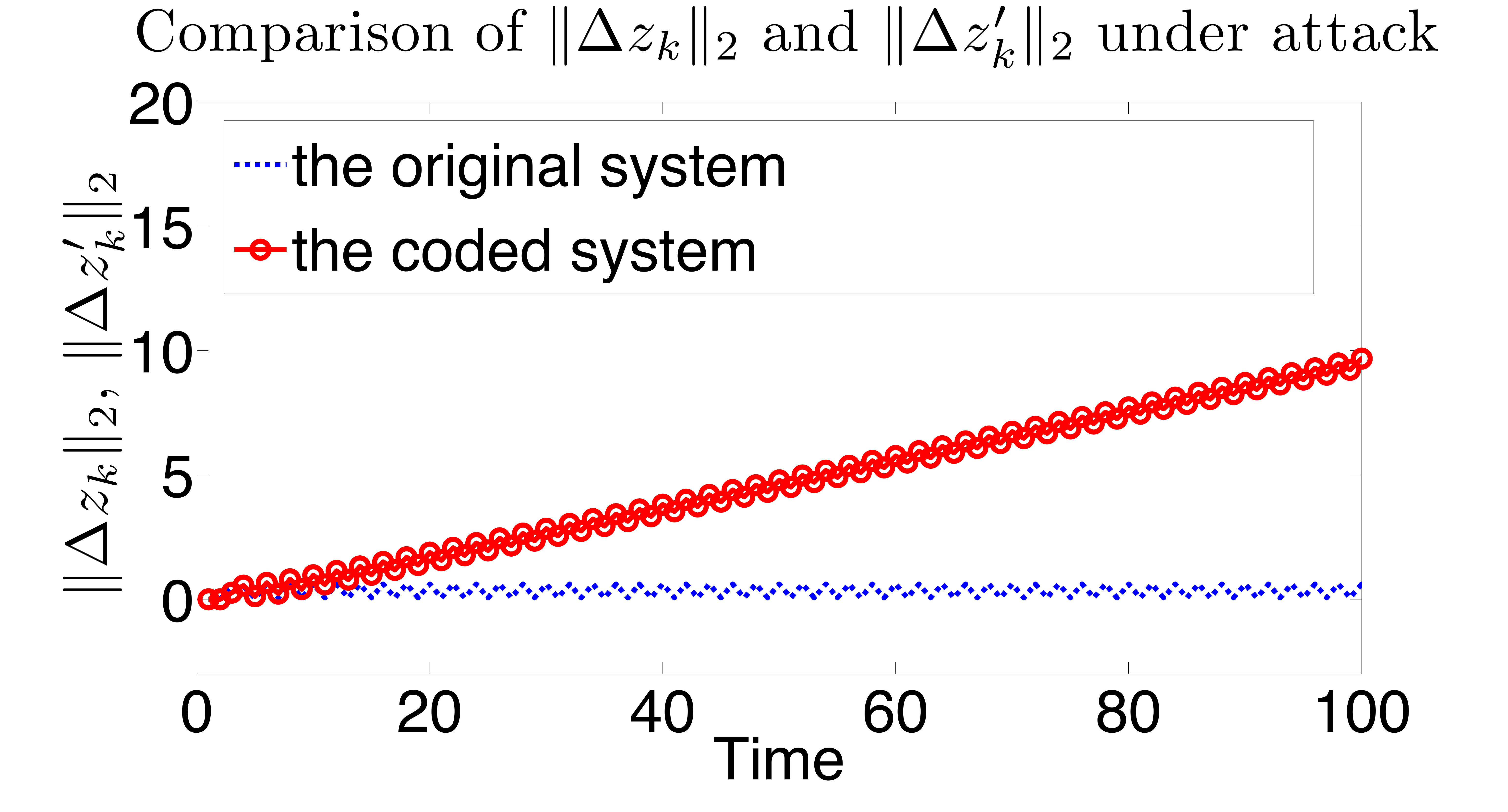}
\vspace{-15pt}
\caption{Comparison of norm of residue change between the original system and coded system, $\Delta z_k$ and $\Delta z'_k$, for $\Sigma_2$ that satisfies Theorem~\ref{code}.}
%For the transformed sensor outputs, $\Delta e'_k$ increases slower than the original system under data injection attack.} 
\label{deltaz_1}
\end{figure}
\begin{figure}[t!]
%\vspace{-5pt}
\centering
\includegraphics [width=0.40\textwidth]{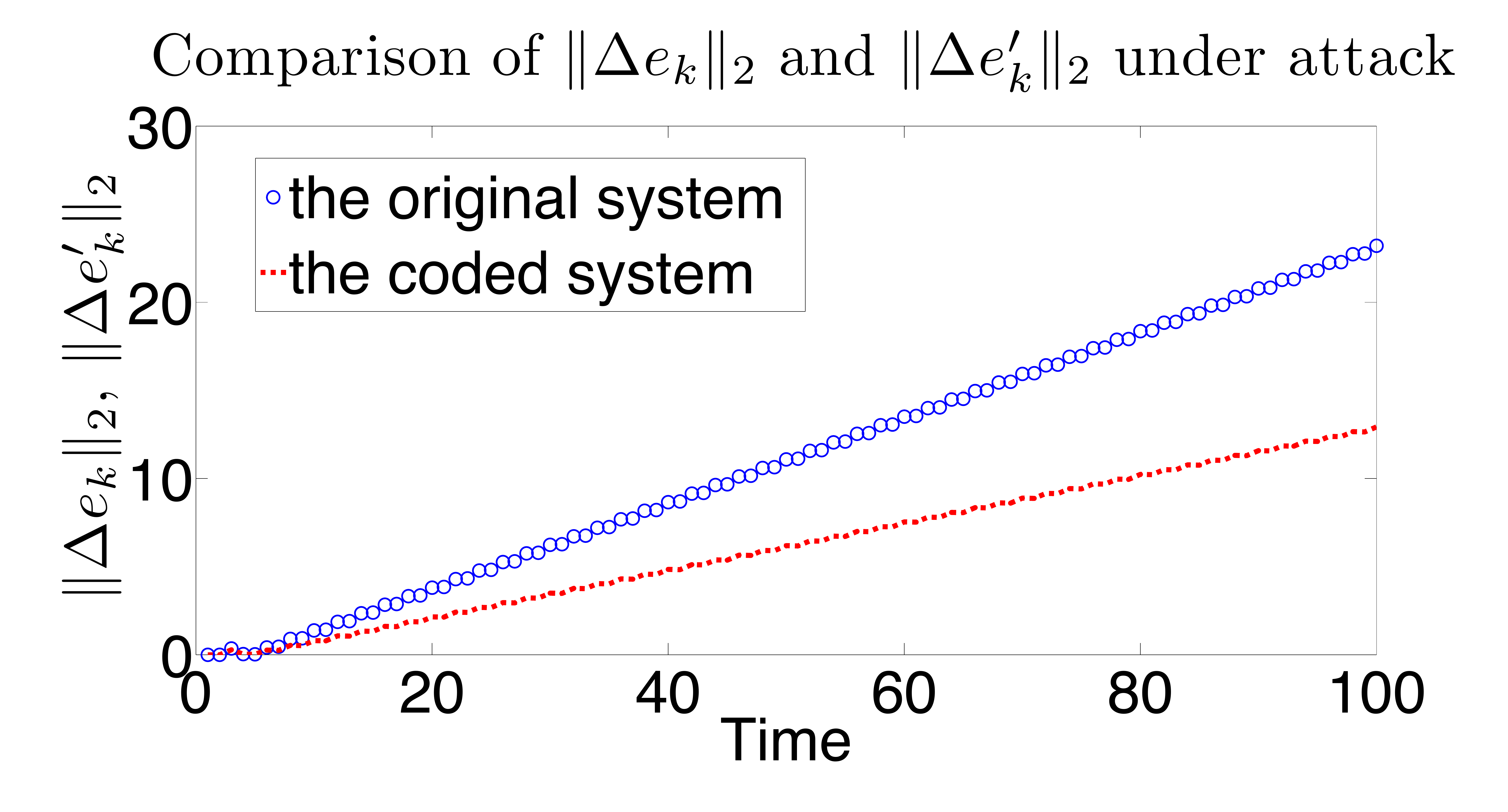}
\vspace{-15pt}
\caption{Comparison of norm of estimation error change between the original system and coded system, 
$\Delta e_k$ and $\Delta e'_k$, for $\Sigma_2$ that satisfies Theorem~\ref{code}.} 
%For the transformed sensor outputs, $\Delta z'_k$ increases with time $k$, while the original system $\Delta z_k$ stays in a bounded range.} 
\label{deltae_1}
\end{figure}
Figure~\ref{deltaz_1} shows that for the sensor outputs transformed by $\Sigma_2$, $\Delta z'_k$ increases with time $k$, while the original system $\Delta z_k$ stays inside a bounded range. For the transformed sensor outputs, the change of the estimation error $\Delta e'_k$ increases even slower than $\Delta e_k$ under data injection attack as shown in Figure~\ref{deltae_1}. By comparing the change of estimation error $\Delta e_k$ and $\Delta e'_k$, we show that estimation error of a coded system does not necessarily increase faster than the original system.
%This is because $K'$ is a steady state Kalman filter gain matrix such that $(A-K'\Sigma CA)$ is stable.

\subsection{When the attacker tries to estimate the coding matrix}
In this example, the system applies the coding matrix designed based on Algorithm~\ref{calculate_g}, a scaled rotation matrix $2*G(1,2, \frac{\pi}{4})$ with a rotation radian $\theta=\frac{\pi}{4}$ in the $(1,2)$ plane 
%\begin{align*}
$\Sigma=\begin{bmatrix}0.7 & 0.5 \\-0.5 & 0.7 \end{bmatrix}.$
%\end{align*}
When the attacker estimates $\Sigma$ according to $N=20$ steps of sensor and actuator measurements via Algorithm~\ref{attacker_alg}, the estimated result is $\hat{\Sigma}$ and the attacker designs a new injection sequence $\hat{\Sigma} y^a_k$ based on $\hat{\Sigma}$
\begin{align*}
%\\\centerline{$
\hat{\Sigma}=\begin{bmatrix}2.80&-0.15\\-0.89&0.05 \end{bmatrix}.%$}
\end{align*}
In Figure~\ref{delta_zs}, we compare the residue change for: the original system under injection sequence $y^a_k$, the coded system under data injection $y^a_k$, and the coded system under injection sequence $\hat{\Sigma} y^a_k$. Assume the threshold for $\|\Delta z_k\|_2$ is set as $M=2$, in Figure~\ref{delta_zs} we can see that the attack will be detected after injecting a sequence of data $y^a_k$ (designed for the original system) for $12$ seconds to the coded system, i.e., $T_s(y^a_k)=12$. In contrast, $T_s(\hat{\Sigma}y^a_k)=50$ seconds, however, $\hat{\Sigma}$ is estimated via $N=20$ seconds of measurements of sensor and actuator values. Hence, the attacker does not have enough time to get such $\hat{\Sigma}$ before being detected.
\begin{figure}[t!]
\vspace{-8pt}
\centering
\includegraphics [width=0.38\textwidth]{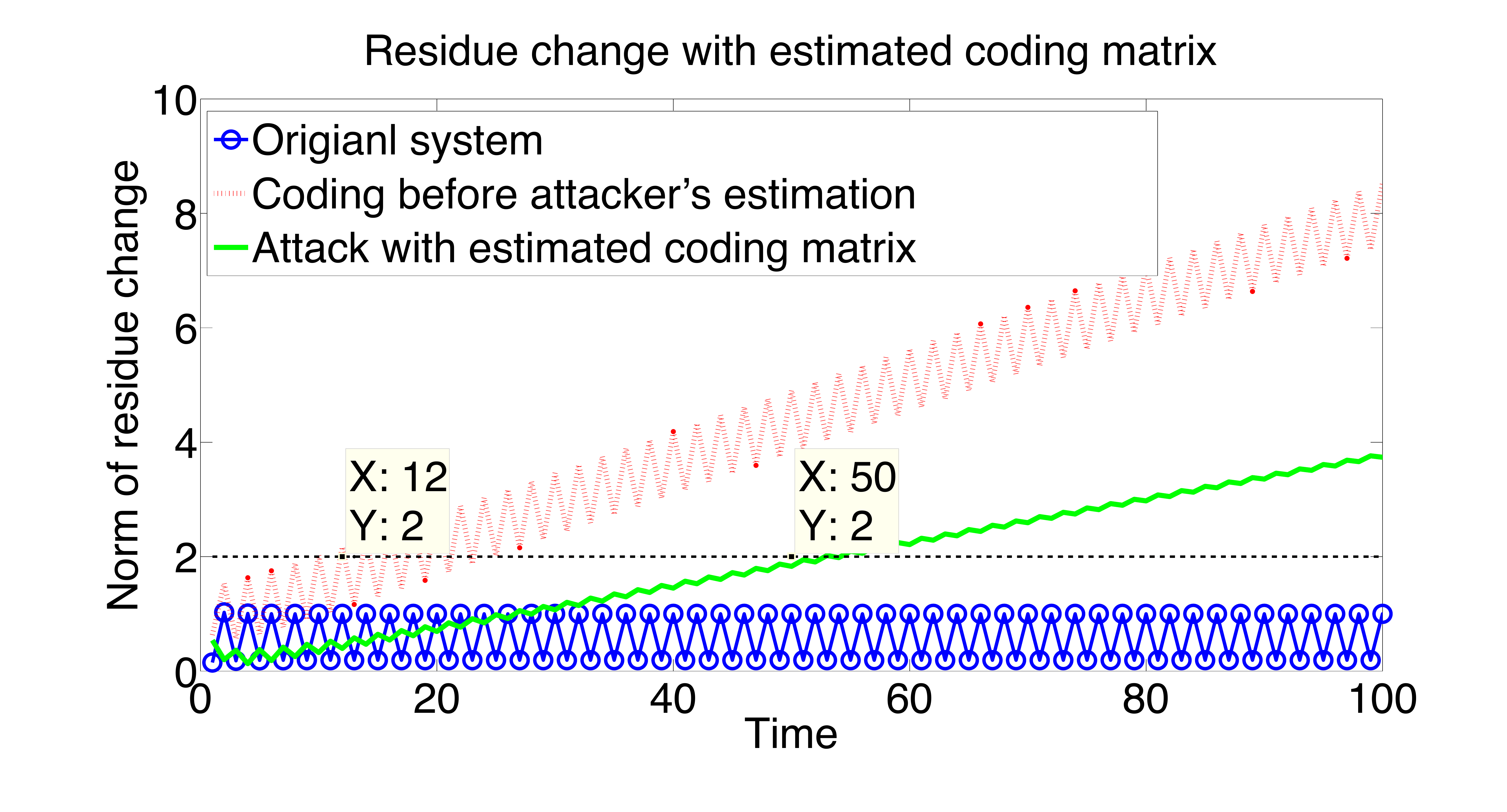}
\vspace{-15pt}
\caption{Comparison of norm of estimation residue change between the original system---$\Delta z_k$,  coded system---$\Delta z'_k$ when the system applies $\Sigma$ and the attacker injects $y^a_k$ designed for the original system, and the sensor data injection sequence designed with estimated coding matrix--- $\Delta z_{sk}$ when the attacker injects $\hat{\Sigma} y^a_k$ and the true coding matrix is $\Sigma$. When $M=2$, $T_s(y^a_k)=12$ seconds, $T_s(\hat{\Sigma}y^a_k)=50$ seconds.} 
%For the transformed sensor outputs, $\Delta z'_k$ increases with time $k$, while the original system $\Delta z_k$ stays in a bounded range.} 
\label{delta_zs}
\end{figure}

\subsection{Number of measurements to estimate the coding matrix}
Figure~\ref{normzs_N} shows how the estimation of $\hat{\Sigma}$ changes with the number measurement steps $N$. In general, when $N$ increases, the difference between $\hat{\Sigma}$ and $\Sigma$ decreases, and the norm of residue change $\|\Delta z_k\|_2$ increases slower with sensor injection sequence $\hat{\Sigma}y^a_k$. However, as shown in Figure~\ref{normzs_N}, for both $N=25$ and $N=200$, $\|\Delta z_k\|_2$ are almost the same, hence, the attacker does not infer a better coding matrix to keep stealthy with a greater measurement time. Comparing the time of keeping stealthy with estimated coding matrix, we have $T_s(\hat{\Sigma}y^a_k)=20$ for $N=2$, $T_s(\hat{\Sigma}y^a_k)=30$ for $N=5$, and approximately $T_s(\hat{\Sigma}y^a_k)=51$ for $N\geqslant 25$. From the perspective of the system, if we set the threshold $\tilde{\alpha} (N_{\Sigma})=1.5$ in this case, and $\frac{T_s( \hat{\Sigma}(N)y^a_k)-T_s(y^a_k)}{T_s (y^a_k)}=\frac{30-12}{12}=1.5$, by the heuristic Algorithm~\ref{change_time}, the system can change the coding matrix every $5$ seconds.

\begin{figure}[t!]
\vspace{-8pt}
\centering
\includegraphics [width=0.38\textwidth]{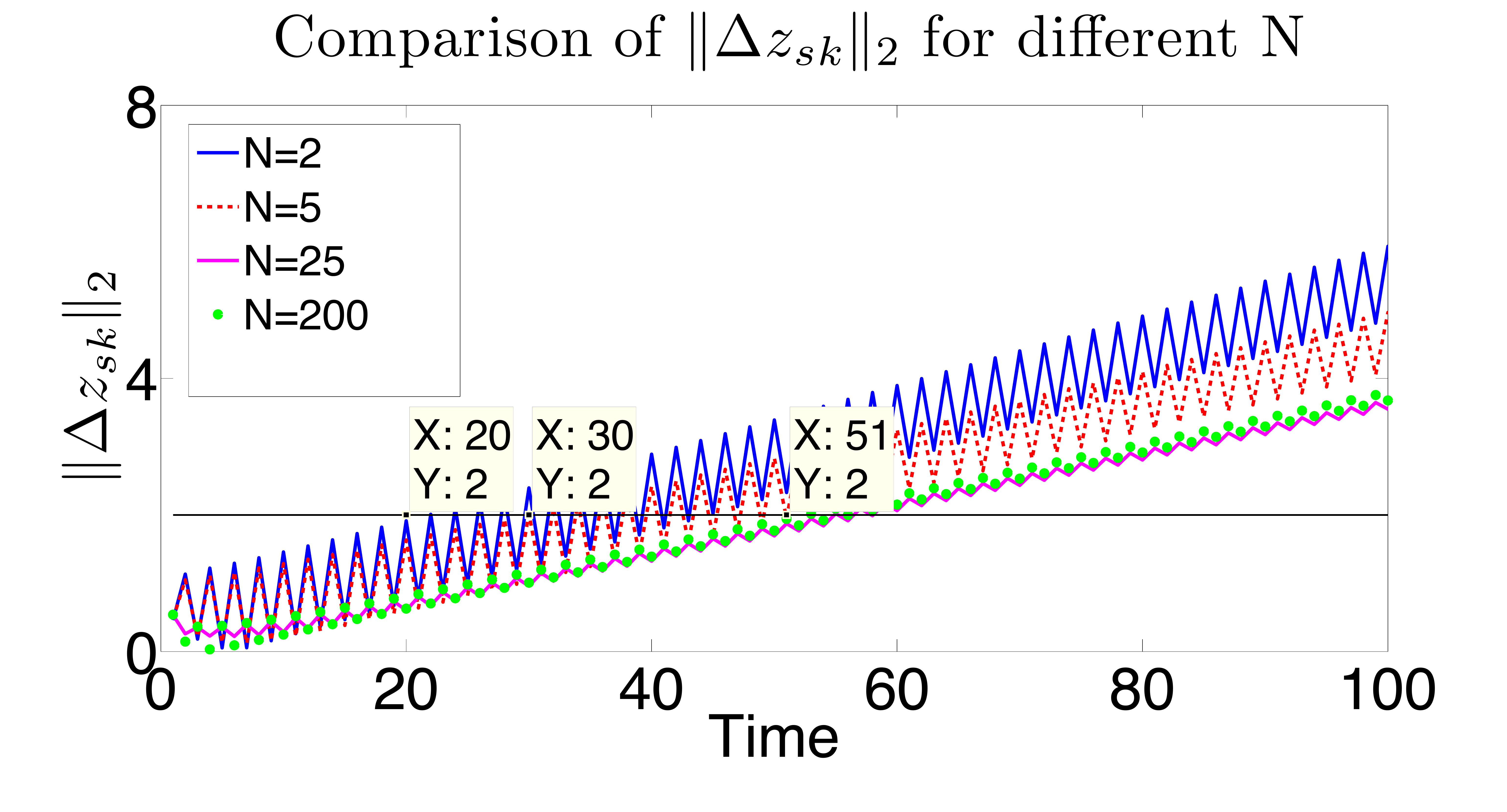}
\vspace{-15pt}
\caption{Comparison of norm of estimation residue change when the attacker designs a sensor injection sequence according to  $\hat{\Sigma}$ estimated with different measurement number $N$. When attacker injects $\hat{\Sigma}y_k^a$ and the system applies coding, the detection time, i.e., the time $\|\Delta z_{sk}\|_2 \geqslant 2$ is labeled for different measurement time $N$: $N=2,T_s(\hat{\Sigma}y^a_k)=20$; $N=5, T_s(\hat{\Sigma}y^a_k)=30$. For $N=25$ and $N=200$, $T_s(\hat{\Sigma}y^a_k)$ are almost the same value: $51$.}
%More steps of observation provides a better estimation that the estimation residue change increases slower. However, after $N=25$, more measurements does not means better injection sequence anymore.} 
%For the transformed sensor outputs, $\Delta z'_k$ increases with time $k$, while the original system $\Delta z_k$ stays in a bounded range.} 
\label{normzs_N}
\vspace{-10pt}
\end{figure}

\section{Conclusion}
\label{sec:conclusion}
In this work, we have proposed a method of coding sensor outputs to detect stealthy data injection attacks that designed by an intelligent attacker with system model knowledge. We show the conditions of a feasible coding scheme to detect a stealthy injection sequence with statistical detectors, and develop an efficient algorithm to compute such feasible coding  matrices. The sensor coding scheme is valid for the scenarios where the attacker is capable to estimate the coding matrix via measuring sensor outputs and actuator inputs. Simulation examples show that the adaptive injection sequence designed based on an  estimated coding matrix cannot pass the detector without knowledge of the coding matrix applied by the system in general. In the future, we will explore a coding scheme for a system with structural constraints.
%and the system needs a new coding matrix before the attacker knows the exact coding.
%In the future, we plan to explore the use of random selection techniques (like game theoretic approaches) to apply a set of feasible transform matrices, such that even the attacker knows the set of transform matrices the system is applying, there is a probability that the injection sequence is not stealth to the current sensor outputs, and the attack is possible to be detected by the system. 

%\input{appendix2}
\bibliographystyle{IEEEtran}
{  \small 
\bibliography{SC}
}

\begin{IEEEbiography}[{\includegraphics[width=1in,height=1.25in,clip,keepaspectratio]{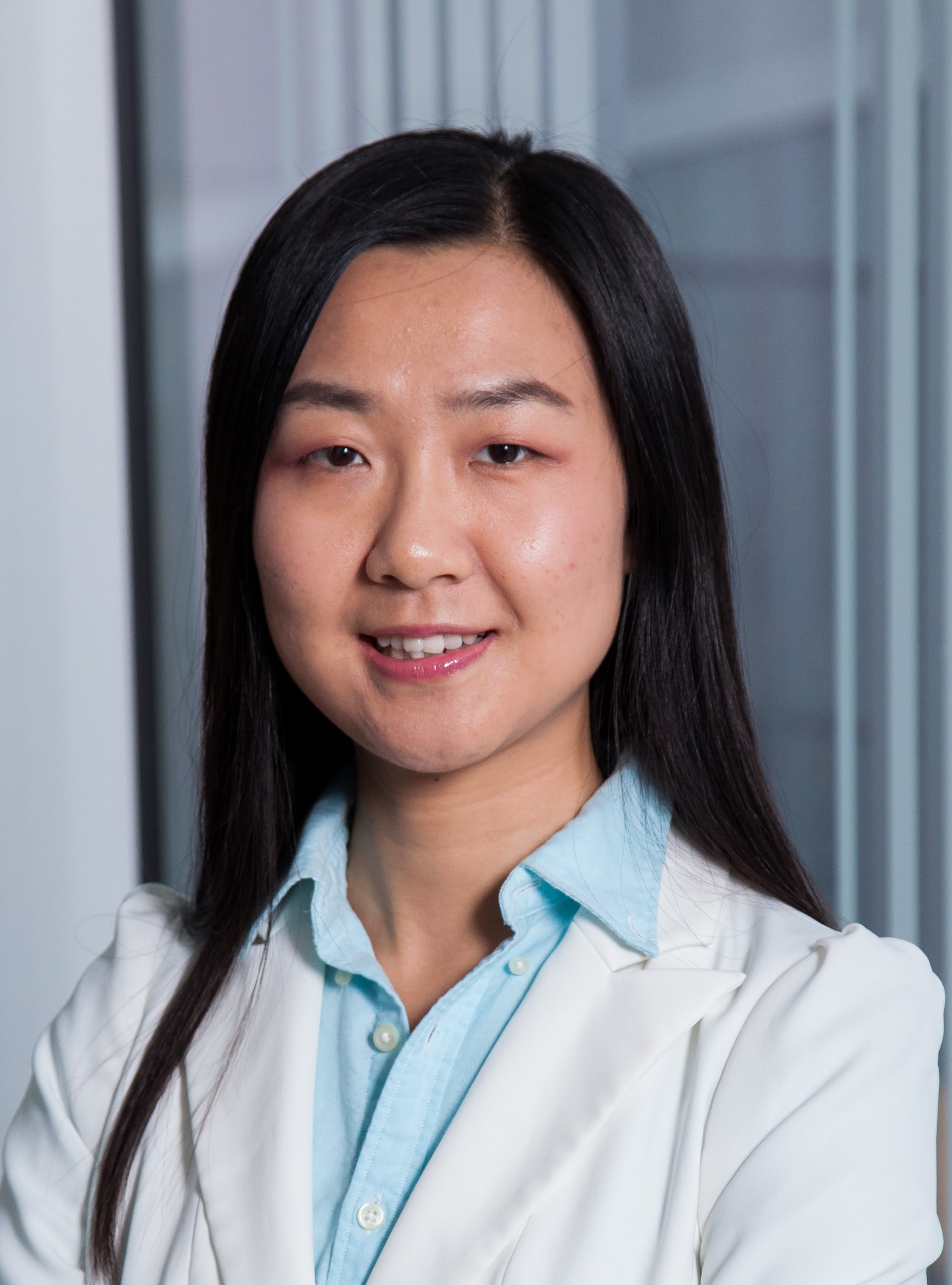}}]{Fei Miao} (S'13) received the B.Sc. degree in Automation from Shanghai Jiao Tong University, Shanghai, China in 2010. Currently, she is working toward the Ph.D. degree in the Department of Electrical and Systems Engineering at University of Pennsylvania. Her research interests include data-driven real-time control frameworks of large-scale interconnected cyber-physical systems under model uncertainties, and resilient control frameworks to address security issues of cyber-physical systems. She was a Best Paper Award Finalist at the 6th ACM/IEEE International Conference on Cyber-Physical Systems in 2015.
\end{IEEEbiography}

\begin{IEEEbiography}[{\includegraphics[width=1in,height=1.25in,clip,keepaspectratio]{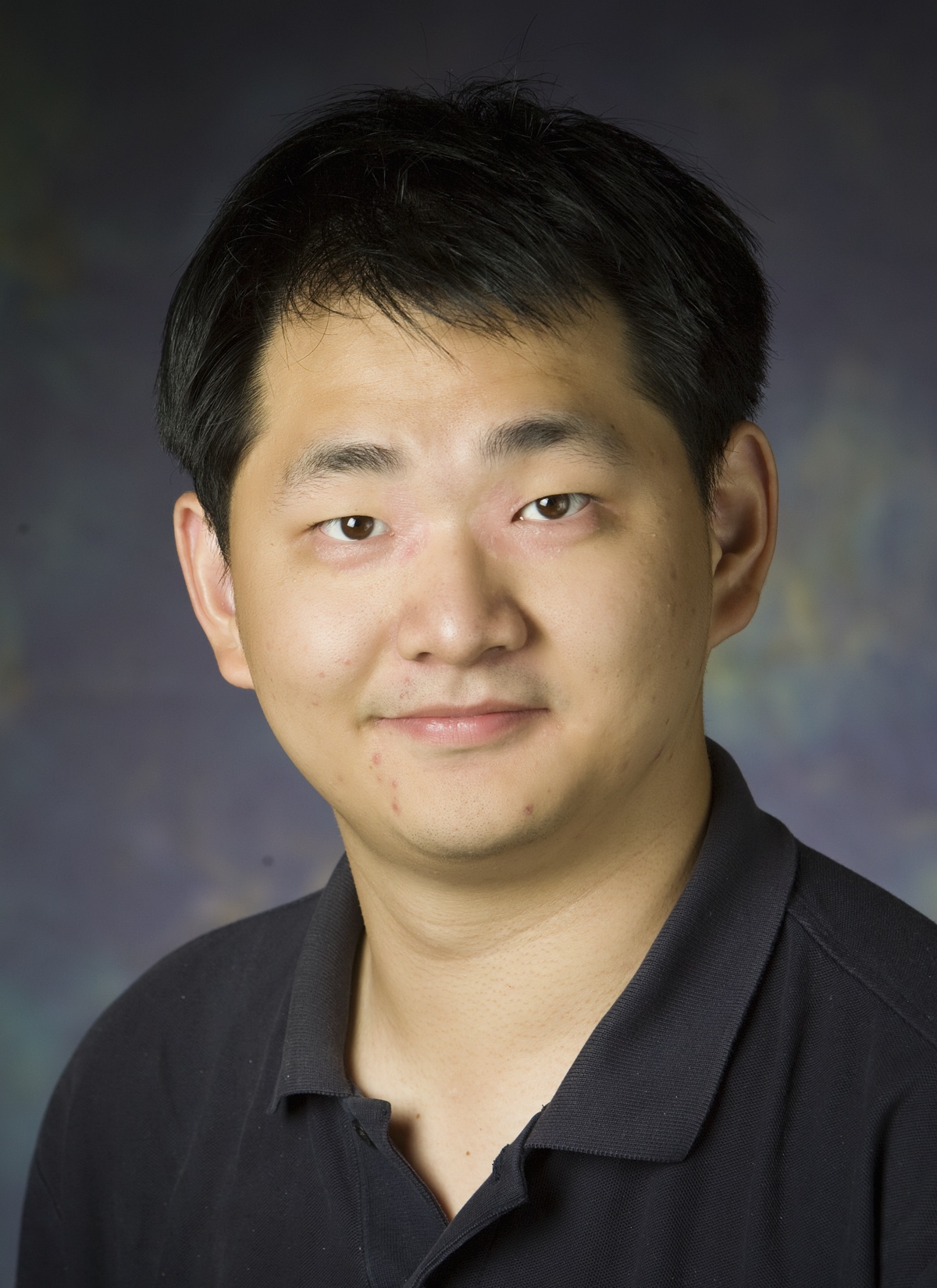}}]{Quanyan Zhu} (S'04-M'12) is an assistant professor in the Department of Electrical and Computer Engineering at New York University. He received the B. Eng. in Honors Electrical Engineering with distinction from McGill University in 2006, the M.A.Sc. from University of Toronto in 2008, and the Ph.D. from the University of Illinois at Urbana-Champaign (UIUC) in 2013. From 2013-2014, he was a postdoctoral research associate at the Department of Electrical Engineering, Princeton University. He is a recipient of many awards including NSERC Canada Graduate Scholarship (CGS), Mavis Future Faculty Fellowships, and NSERC Postdoctoral Fellowship (PDF). He spearheaded and chaired INFOCOM Workshop on Communications and Control on Smart Energy Systems (CCSES), Midwest Workshop on Control and Game Theory (WCGT), and 7th Game and Decision Theory for Cyber Security (GameSec). His current research interests include resilient and secure interdependent critical infrastructures, energy systems, cyber-physical systems, and smart cities.
\end{IEEEbiography}

\begin{IEEEbiography}[{\includegraphics[width=1in,height=1.25in,clip,keepaspectratio]{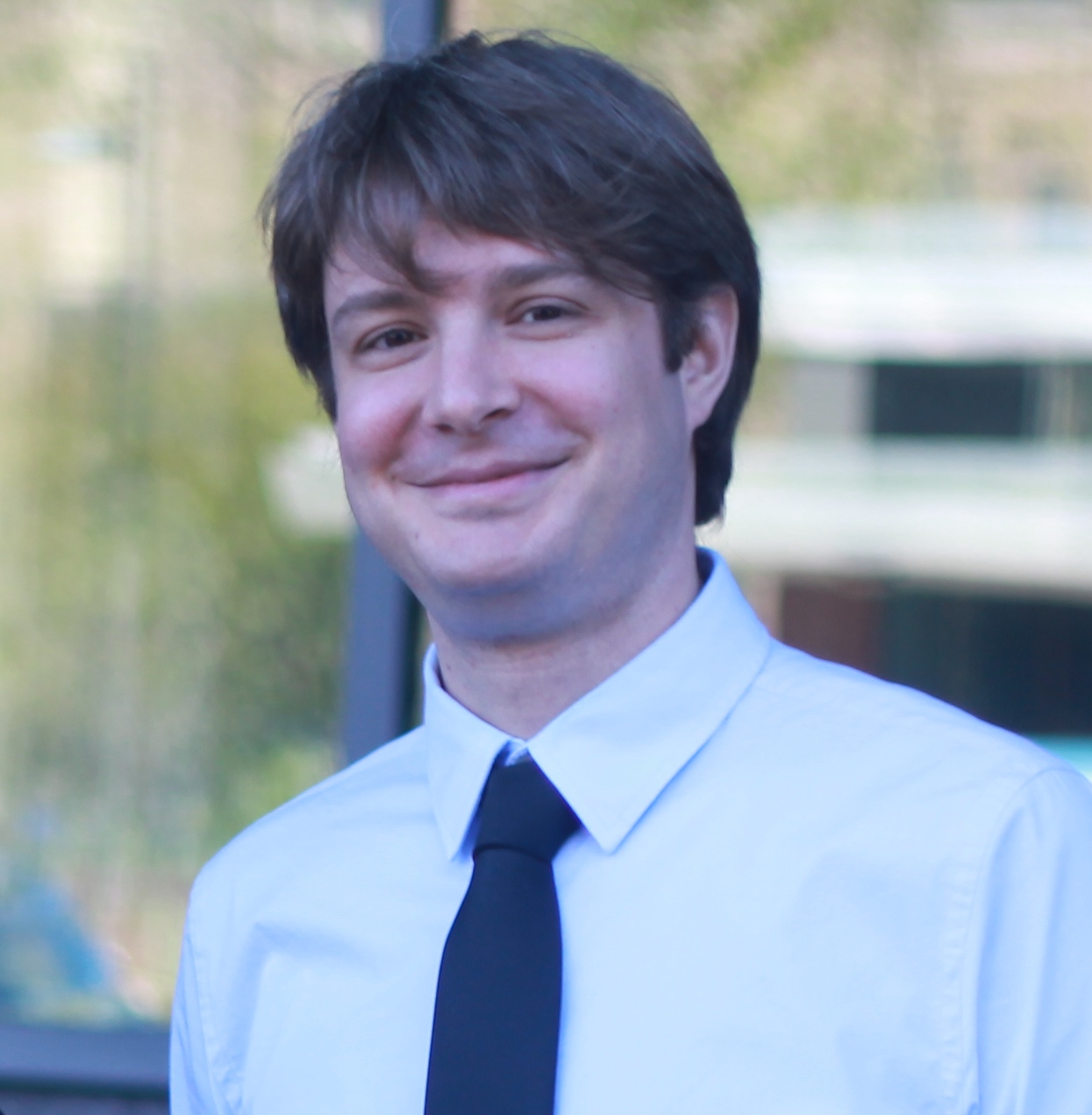}}]{Miroslav Pajic} (S’06-M’13) received the Dipl. Ing. and M.S. degrees in electrical engineering from the University of Belgrade, Serbia, in 2003 and 2007, respectively, and the M.S. and Ph.D. degrees in electrical engineering from the University of Pennsylvania, Philadelphia, in 2010 and 2012, respectively. He is currently an Assistant Professor in the Department of Electrical and Computer Engineering at Duke University. He also holds a secondary appointment in the Computer Science Department. Prior to joining Duke, Dr. Pajic was a Postdoctoral Researcher in the PRECISE Center, University of Pennsylvania, from 2012-2015. His research interests focus on the design and analysis of cyber-physical systems and in particular real-time and embedded systems, distributed/networked control systems, and high-confidence medical devices and systems. Dr. Pajic received various awards including the 2011 ACM SIGBED Frank Anger Memorial Award, the Joseph and Rosaline Wolf Award for Best Electrical and Systems Engineering Dissertation from Penn, the Best Paper Award at the 2014 ACM/IEEE International Conference on Cyber-Physical Systems (ICCPS), and the Best Student Paper award at the 2012 IEEE Real-Time and Embedded Technology and Applications Symposium (RTAS).
\end{IEEEbiography}

\begin{IEEEbiography}[{\includegraphics[width=1in,height=1.25in,clip,keepaspectratio]{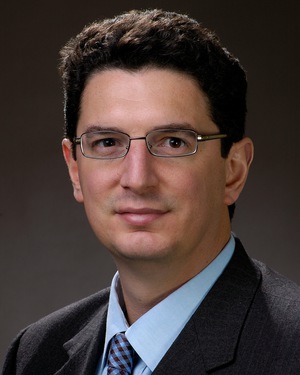}}]{George J. Pappas} (S'90-M'91-SM'04-F'09) received the Ph.D. degree in electrical engineering and computer sciences from the University of California, Berkeley, CA, USA, in 1998. He is currently the Joseph Moore Professor and Chair of the Department of Electrical and Systems Engineering, University of Pennsylvania, Philadelphia, PA, USA. He also holds a secondary appointment with the Department of Computer and Information Sciences and the Department of Mechanical Engineering and Applied Mechanics. He is a Member of the GRASP Lab and the PRECISE Center. He had previously served as the Deputy Dean for Research with the School of Engineering and Applied Science. His research interests include control theory and, in particular, hybrid systems, embedded systems, cyber-physical systems, and hierarchical and distributed control systems, with applications to unmanned aerial vehicles, distributed robotics, green buildings, and bimolecular networks. Dr. Pappas has received various awards, such as the Antonio Ruberti Young Researcher Prize, the George S. Axelby Award, the Hugo Schuck Best Paper Award, the George H. Heilmeier Award, the National Science Foundation PECASE award and numerous best student papers awards at ACC, CDC, and ICCPS.
\end{IEEEbiography}

\end{document}